\documentclass[10pt]{article}
\usepackage[utf8]{inputenc}
\usepackage[left=0.81in, right=0.81in, top=1in, bottom=1.6in]{geometry}
\usepackage{xcolor}
\usepackage{amsmath} 
\usepackage{amsfonts}
\usepackage{graphicx}
\usepackage{multirow}
\usepackage{amsthm}
\usepackage{amssymb}
\usepackage{natbib}
\usepackage{hyperref}
\usepackage{bookmark}

\newcommand{\N}{\mathbb N} 
\newcommand{\R}{\mathbb R} 
\newcommand{\simplex}{\mathbb S} 

\newcommand{\PS}{\mathcal{P}} 
\newcommand{\QQ}{\mathcal{Q}}

\newcommand{\I}{\mathcal{I}}
\newcommand{\T}{\mathcal{T}}

\newcommand{\blambda}{\boldsymbol{\lambda}}
\newcommand{\bnu}{\boldsymbol{\nu}}
\newcommand{\bq}{\boldsymbol{q}}
\newcommand{\bn}{\boldsymbol{n}}
\newcommand{\bx}{\boldsymbol{x}}
\newcommand{\bh}{\boldsymbol{h}}
\newcommand{\bg}{\boldsymbol{g}}
\newcommand{\bpi}{\boldsymbol{\pi}}
\newcommand{\bmm}{\boldsymbol{m}}
\newcommand{\by}{\boldsymbol{y}}
\newcommand{\bv}{\boldsymbol{v}}
\newcommand{\bone}{\boldsymbol{1}}

\newtheorem{lemma}{Lemma}
\newtheorem{defi}{Definition}
\newtheorem{theorem}{Theorem}
\newtheorem{case}{Case}

\title{The Central Role of the Identifying Assumption in Population Size Estimation}
\author{Serge Aleshin-Guendel$^{1}$,  Mauricio Sadinle$^{1}$, and Jon Wakefield$^{1, 2}$ \\ $^{1}$Department of Biostatistics, University of Washington, Seattle, Washington, U.S.A. \\ $^{2}$Department of Statistics, University of Washington, Seattle, Washington, U.S.A.}
\date{}

\begin{document}

\maketitle

\begin{abstract}
 The problem of estimating the size of a population based on a subset of individuals observed across multiple data sources is often referred to as capture-recapture or multiple-systems estimation. This is fundamentally a missing data problem, where the number of unobserved individuals represents the missing data. As with any missing data problem, multiple-systems estimation requires users to make an untestable identifying assumption in order to estimate the population size from the observed data. If an appropriate identifying assumption cannot be found for a data set, no estimate of the population size should be produced based on that data set, as models with different identifying assumptions can produce arbitrarily different population size estimates---even with identical observed data fits.  Approaches to multiple-systems estimation often do not explicitly specify identifying assumptions. This makes it difficult to decouple the specification of the model for the observed data from the identifying assumption and to provide justification for the identifying assumption.
 We present a re-framing of the multiple-systems estimation problem that leads to an approach which decouples the specification of the observed-data model from the identifying assumption, and discuss how common models fit into this framing. This approach takes advantage of existing software and facilitates various sensitivity analyses.
 We demonstrate our approach in a case study estimating the number of civilian casualties in the Kosovo war.
 Code used to produce  this manuscript is available at \href{https://github.com/aleshing/central-role-of-identifying-assumptions}{\texttt{github.com/aleshing/central-role-of-identifying-assumptions}}.
 \\
 Keywords: Capture-recapture; Missing data; Multiple-systems estimation; Sensitivity analysis.
\end{abstract}


%

\section{Introduction}
Estimating the size of a closed population is a common problem in many fields, including ecology \citep{Otis_1978}, epidemiology \citep{Hook_1995a}, official statistics \citep{ Anderson_1999}, and human rights \citep{Ball_2002}. The available data typically take the form of multiple lists which record information on a subset of individuals in a population. When there exists a mechanism to identify which individuals are the same across lists, multiple-systems estimation (MSE), also known as capture-recapture, provides an approach to estimating the population size based on the overlap of the lists \citep{Bird_2018}.

MSE is at its heart a missing data problem, as we do not observe all individuals in the population of interest  \citep[see e.g.][]{Fienberg_2009, Manrique-Vallier_2016}. As in any missing data problem, MSE requires users to make an untestable identifying assumption about how the observed individuals relate to the unobserved individuals in order to estimate the population size from the observed data. In practice, this means that models with different identifying assumptions can produce arbitrarily different population size estimates, even when the models have identical fits to the observed data. Thus, any identifying assumptions used in an analysis need to be appropriately justified based on the context of the data. If an appropriate identifying assumption can not be found for a data set, no estimate of the population size should be produced based on that data set.

We believe that the central role of specifying the identifying assumption is not sufficiently appreciated, as it is usually conflated with model specification, which involves both making an identifying assumption \textit{and} specifying a model for the observed data. See for example \cite{Fienberg_1972} who wrote ``... we are assuming that the model which describes the observed data also describes the count of the unobserved individuals. We have no way of checking this assumption," and \cite{Manrique-Vallier_2013} who wrote ``The arguably most basic assumption in MSE is that the noninclusion of the fully unobserved individuals ... can be represented by the same model that represents the inclusion (and noninclusion) of those we can observe in at least one list. This is a strong and untestable condition." 

This conflation of identifying assumption specification and model specification has led practitioners to perform model evaluation by comparing a suite of model fits that are the results of both fundamentally different identifying assumptions and different model specifications for the observed data \citep[see e.g.][]{Sadinle_2018, Manrique-Vallier_2019b, Silverman_2020}. This makes it essentially impossible to disentangle whether differences in inferences are due to differences in identifying assumptions, model specifications for the observed data, or some combination. More importantly, it is rare in these instances for practitioners to provide justification for any of the identifying assumptions being used.

In this article, we propose an approach for MSE that places the identifying assumption front and center in the MSE workflow. We first revisit the framing of MSE as a missing data problem and describe our approach in Section \ref{sec:msemiss}. Section \ref{sec:common_models} reviews two common MSE models\textemdash log-linear and latent class models\textemdash through our missing data framing. In Section \ref{sec:revisit} we focus on the identifying assumption associated with log-linear models, and describe how it can be used as a building block for alternative identifying assumptions and sensitivity analyses that examine the impact of the identifying assumption. 
Finally, in Section \ref{sec:kosovo} we illustrate our approach in a case study of estimating the number of civilian casualties in the Kosovo war.

\section{Multiple-Systems Estimation as a Missing Data Problem}
\label{sec:msemiss}
\subsection{The Data}
Suppose we have a closed population of $N$ individuals, of which $n<N$ are observed by one or more of $K$ lists. Let $H=\{0, 1\}^K$ denote the possible patterns of inclusion of the individuals in the lists, $H^*=H\setminus \{0\}^K$ denote the possible subsets of lists in which each of the $n$ observed individuals could have been observed, and let $\bx_i\in H$ denote the subset of lists in which individual $i$ was included. For example, with $K=3$, $\bx_i=(0, 1, 1)$ indicates that individual $i$ was observed in lists $2$ and $3$, but not list $1$. 
	
These data for the $N$ individuals can be gathered into a $2^K$ contingency table of list overlap, where the cells of the table are indexed by $\bh\in H$, with counts $n_{\bh}=\sum_{i=1}^NI(\bx_i=\bh)$. We do not observe the count for cell $\{0\}^K$, $n_0:= n_{(0,\cdots,0)} = N-n$, which records the number of individuals missing from all lists, so the observed contingency table is incomplete. Let $\bn =\{n_{\bh}\}_{{\bh}\in H^*}$ denote the counts of the incomplete contingency table. The unobserved cell count $n_0$, or equivalently the population size $N$, is the target of inference.

\subsection{The Complete-Data Distribution}
\label{sec:comp_data_dist}
Under independent and identically distributed (i.i.d.) sampling of individuals by the lists, the $2^K$ contingency table of counts is multinomially distributed, i.e.
\begin{equation}
\label{eq:working_model}
    \bn, n_0 \mid N, \bpi \sim \textsc{Multinomial}(N, \bpi),
\end{equation}
where $\bpi=\{\pi_{\bh}\}_{{\bh}\in H} \in \simplex^{2^K-1}$ is a set of cell probabilities, and $\simplex^{d}=\{(a_1,\cdots, a_{d+1})\in\R^{d+1}\mid \sum_{i=1}^{d+1}a_i=1, a_i> 0 \ \forall i\}$ denotes the $d$-dimensional probability simplex. {We note that this multinomial model, introduced as early as \cite{Darroch_1958}, is a possible simplification of reality, as it does not allow for correlation of individuals' inclusion patterns.} We will refer to the model in \eqref{eq:working_model} as the \textit{complete-data distribution}, for which the evaluation relies on knowing the complete $2^K$ contingency table of counts. In general, the parameter space for this model will be some subset of $\Theta=\{N, \bpi \mid N\in\N, \bpi\in \mathbb{S}^{2^K-1}\}$, which we will refer to as the \textit{complete-data parameterization}. As shown in Web Appendix A, when individuals are not i.i.d. sampled, but are sampled independently with cell probabilities drawn i.i.d. from some mixing distribution on $\mathbb{S}^{2^K-1}$, we also arrive at the model in \eqref{eq:working_model}. This is the case for common models for heterogeneity such as the $M_h$ and $M_{th}$ models of \cite{Otis_1978}. {Because common models for heterogeneity reduce to \eqref{eq:working_model}, in the rest of this article we will view the cell probabilities as being marginal of any possible heterogeneity mechanisms.}
	
\subsection{Decomposing the Complete-Data Distribution}


It is instructive to decompose the complete-data distribution as
\begin{equation}\label{eq:decomp}
    p(\bn, n_0\mid  N, \bpi)= N!\prod_{\bh\in H}\frac{\pi_{\bh}^{n_{\bh}}}{n_{\bh}!}= L_1(N, \pi_0\mid n) L_2(\tilde{\bpi}\mid \bn),
\end{equation}
with $L_1(N, \pi_0\mid n)=\binom{N}{n} \pi_0^{N-n} (1-\pi_0)^{n}$ 
and $L_2(\tilde{\bpi}\mid \bn)= n!\prod_{\bh\in H^*}\tilde{\pi}_{\bh}^{n_{\bh}}/n_{\bh}!$, where $\pi_0:=\pi_{(0,\cdots,0)}=1-\sum_{\bh\in H^*}\pi_{\bh}$ is the probability of being missing from every list, and $\tilde{\pi}_{\bh}=\frac{\pi_{\bh}}{1-\pi_0}=\frac{\pi_{\bh}}{\sum_{\bh'\in H^*}\pi_{\bh'}}$ is the probability of being observed in the subset of the lists $\bh$ conditional on being observed in at least one list. 
$L_1$ is a binomial likelihood for $n$, which has been well studied in the related binomial $N$ problem literature \citep[see e.g. ][]{Rukhin_1975}. $L_2$ is a multinomial likelihood for the observed data $\bn$ conditional on their sum $n$, referred to as the \textit{conditional likelihood} \citep{Fienberg_1972}. We will refer to $\pi_0$ as the \textit{unobserved cell probability} and to $\tilde{\bpi}$ as the \textit{observed cell probabilities}. This decomposition hints at an alternative to the complete-data parameterization $\Theta$, $\Theta^*=\{N, \pi_0, \tilde{\bpi}\mid N\in\N, \pi_0\in(0,1), \tilde{\bpi}\in \mathbb{S}^{2^K-2}\}$, which we will refer to as the \textit{observed-data parameterization}. The two parameterizations are equivalent, so we will work with whichever is more convenient for exposition.

\subsection{Identifiability} 
\label{sec:ident}
Before performing inference in a statistical model, it is important to check that the model is identifiable. For $\theta\in\Theta^*$, let $P_\theta$ denote the complete-data distribution at the set of parameters $\theta$. Consider the following standard definition of identifiability:
\begin{defi}
	The statistical model $\PS_{\Omega}=\{P_{\theta}\mid  \theta\in\Omega\subset\Theta^*\}$ is \textbf{identifiable} if $\ \forall \theta_1,\theta_2\in\Omega$, $P_{\theta_1}=P_{\theta_2}$ implies that $\theta_1=\theta_2$. Equivalently, $\PS_{\Omega}$ is identifiable if $\ \forall\theta_1=\{N, \pi_0, \tilde{\bpi}\},\theta_2=\{N', \pi_0', \tilde{\bpi}'\}\in\Omega$, $L_1(N, \pi_0\mid n)L_2(\tilde{\bpi}\mid \bn)=L_1(N', \pi_0'\mid n)L_2(\tilde{\bpi}'\mid \bn)$ $\forall\bn$ implies that $\theta_1=\theta_2$.
\end{defi}
One can show that the unrestricted model $\PS_{\Theta^*}$ is identifiable. Since the goal is to estimate $N$, sufficiency might lead one to try to estimate $N$ and $\pi_0$ in the unrestricted model based solely on the binomial likelihood for $n$. Examining the likelihood surface for a given $n$, one finds a maximum at $N=n$ and $\pi_0=0$, with a ridge centered along the set $\{N\in\N, \pi_0\in(0,1)\mid N(1-\pi_0)\approx n\}$ that monotonically decreases as $N$ increases. In Figure \ref{fig:L1surface} we plot this surface when $n=100$. There is a fundamental problem in that two parameters are being estimated with one data point, which makes it impossible to construct an unbiased or consistent estimator of either $N$ or $\pi_0$ \citep{DasGupta_2005, Farcomeni_2012}. Thus the standard definition of identifiability is misleading in this setting, as it does not necessarily imply that the parameters are estimable in any traditional sense.

\begin{figure}[h]
\centering
\includegraphics[width=0.75\linewidth]{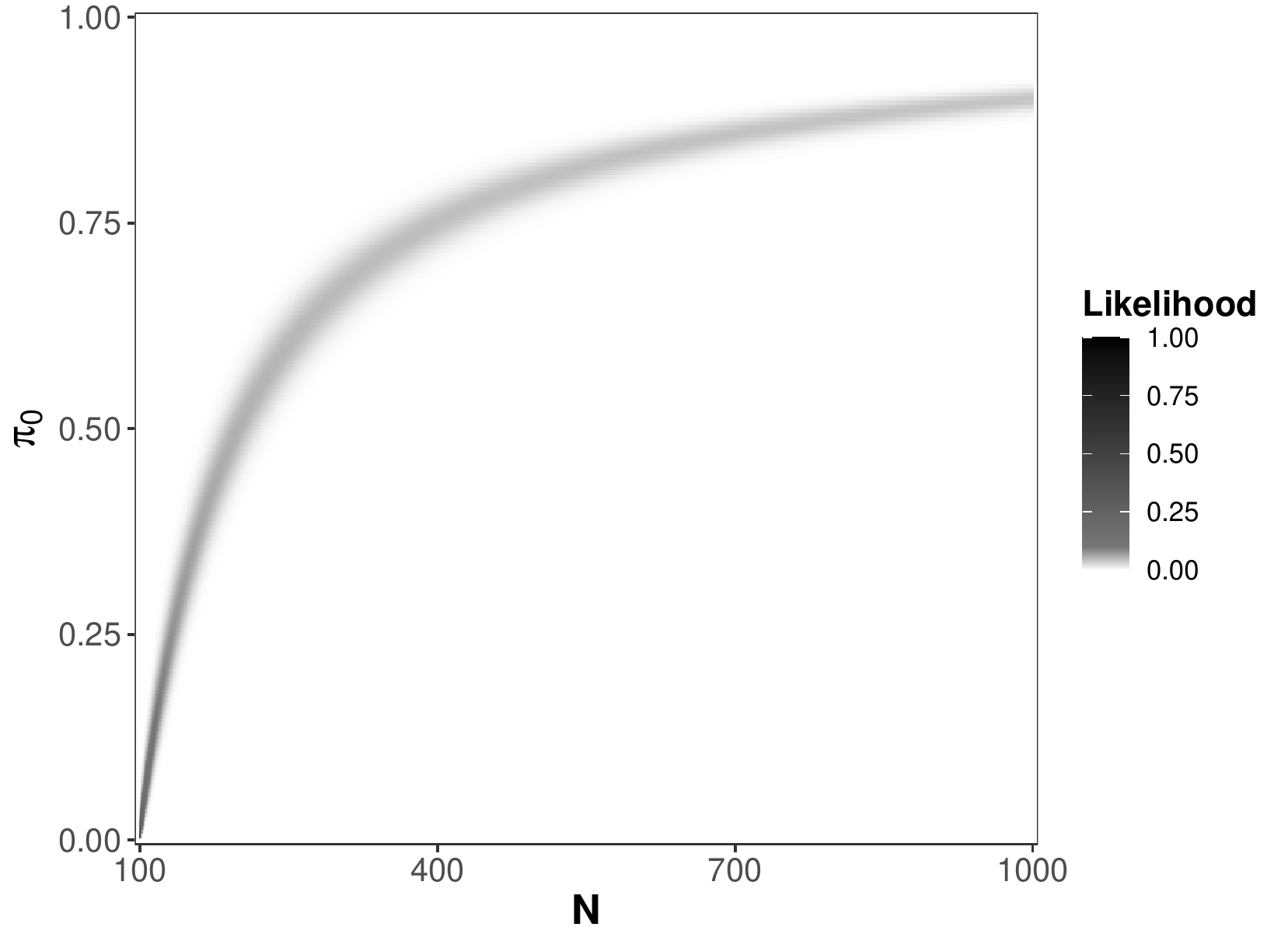}
\caption{Likelihood surface of $L_1$ when $n=100$.}
\label{fig:L1surface}
\end{figure}

We will instead use the following alternative definition of identifiability specific to MSE \citep{Link_2003, Holzmann_2006}:
\begin{defi}
    \label{def:cond_indet}
	The statistical model $\PS_{\Omega}$ is \textbf{conditionally identifiable} if $\ \forall\theta_1=\{N, \pi_0, \tilde{\bpi}\},$ $\theta_2=\{N', \pi_0', \tilde{\bpi}'\}\in\Omega$, $L_2(\tilde{\bpi}\mid \bn)=L_2(\tilde{\bpi}'\mid \bn)$ $\forall\bn$ implies that $ \pi_0= \pi_0'$. 
\end{defi}
In a conditionally identifiable model, the conditional likelihood, $L_2$, identifies the unobserved cell probability, $\pi_0$. Clearly the unrestricted model $\PS_{\Theta^*}$ is not conditionally identifiable. 
Standard identifiability of the multinomial conditional likelihood tells us that we can equivalently state Definition \ref{def:cond_indet} as follows: the statistical model $\PS_{\Omega}$ is conditionally identifiable if $\ \forall\theta_1=\{N, \pi_0, \tilde{\bpi}\},\theta_2=\{N', \pi_0', \tilde{\bpi}'\}\in\Omega$, $\tilde{\bpi}=\tilde{\bpi}'$ implies that $ \pi_0= \pi_0'$.
Thus for a conditionally identifiable model, there exists a function $\T\colon \tilde{T}\to(0,1)$ that maps observed cell probabilities, $\tilde{\bpi}$, to unobserved cell probabilities, $\pi_0$, where $\tilde{T}\subset\mathbb{S}^{2^K-2}$. When the domain $\tilde{T}$ of this function is not equal to $\mathbb{S}^{2^K-2}$, this restricts the set of possible values for $\tilde{\bpi}$ in the model to $\tilde{T}$. Any extra assumptions in the model involving $\tilde{\bpi}$ can then further restrict the set of possible values for $\tilde{\bpi}$ in the model to a set $\tilde{S}\subset\tilde{T}$.
Thus conditionally identifiable models take the form $\PS_{\Omega}$, where $\Omega=\{N, \pi_0, \tilde{\bpi}\mid N\in \N, \pi_0=\T(\tilde{\bpi}), \tilde{\bpi}\in \tilde{S}
\}$.

When a model is not conditionally identifiable, we have no guarantees for when the parameters are estimable in any traditional sense. In particular, non-identifiability precludes consistent estimation as ``there will be uncertainty in parameter estimates that is not washed out as more data are collected" \citep{Linero_2017}. If a model $\PS_{\Omega}$ is conditionally identifiable, all parameters of the model can be consistently estimated \citep{Sanathanan_1972}. However, we emphasize that the data needs to have been generated by a distribution 
in the model $\PS_{\Omega}$ for the parameters to be consistently estimable.  In other words, in order to estimate the population size $N$, we need to assume a functional relationship, $\T$, between the observed cell probabilities $\tilde{\bpi}$ and the unobserved cell probability $\pi_0$. This is the main idea behind MSE.
	
\subsection{Missing Data} \label{sec:inf_miss}
The framing in the previous section is motivated by our treatment of MSE as a missing data problem. The decomposition in \eqref{eq:decomp} is related to the decomposition in the missing data literature of the complete-data distribution into the \textit{extrapolation distribution} and the \textit{observed-data distribution} \citep{Hogan_2008}. The extrapolation distribution captures how to extrapolate to the missing data given the observed data, which in our context corresponds to $L_1$. The observed-data distribution, as the name indicates, is the distribution of the observed data, which in this context corresponds to $L_2$. Following the analogy of the missing data literature, by restricting ourselves to models of the form $\PS_{\Omega}$, where $\Omega=\{N, \pi_0, \tilde{\bpi}\mid N\in \N, \pi_0=\T(\tilde{\bpi}), \tilde{\bpi}\in \tilde{S}\}$, one is making an \textit{identifying assumption}, $\T$,  about how $\tilde{\bpi}$ relates to $\pi_0$ in order to identify $\pi_0$.

The observed-data distribution is restricted when the set of possible values for the observed cell probabilities, $\tilde{S}$, is not equal to $\mathbb{S}^{2^K-2}$. Based on standard properties of the multinomial conditional likelihood, restrictions on the observed-data distribution  are assumptions that are testable from the data.
As noted in the previous section, these restrictions could be due to the domain, $\tilde{T}$, of the identifying assumption (see Section \ref{sec:marg_ident} for an example), or due to extra modeling assumptions for the observed cell probabilities, $\tilde{\bpi}$ (see Section \ref{sec:loglinear} for an example). This motivates the following definition \citep[see Chapter 8 of][]{Hogan_2008}:
\begin{defi}
    A model $\PS_{\Omega}$, where $\Omega=\{N, \pi_0, \tilde{\bpi}\mid N\in \N, \pi_0=\T(\tilde{\bpi}), \tilde{\bpi}\in \tilde{S}\}$, is \textit{nonparametric identified} when $\tilde{S}=\tilde{T}= \mathbb{S}^{2^K-2}$, i.e. the observed-data distribution is not restricted by the model.
\end{defi}

\subsection{Our Approach to Multiple-Systems Estimation}
\label{sec:outlook}

In the MSE literature, previous work has been concerned with determining \textit{when} certain models are conditionally identified \citep[see e.g.][]{Link_2003,Holzmann_2006}. Here we are concerned with determining both when \textit{and how} models are conditionally identified. Since the validity of our inferences rests on the untestable identifying assumption and any restrictions on the observed-data distribution being correct, we would like to know what identifying assumption we are actually making so we can determine whether or not the assumption is plausible in a given context. 
Thus, in this article our approach to MSE will be to use conditionally identified models that are based on \textit{explicitly specified} identifying assumptions. Additionally, to make as few testable assumptions as possible, we will use models where the observed-data distribution is only possibly restricted by the identifying assumption (i.e. $\tilde{S}=\tilde{T}$).

Given such a conditionally identified model
, our approach to MSE is agnostic to the inferential framework used, so one can perform inference for $N$ in a frequentist or Bayesian framework. In Web Appendix B, we outline how computation, including sensitivity analyses probing the identifying assumption as we will describe in Section \ref{sec:revisit}, can be carried out in either framework using existing software. 

In the rest of this article, we examine the identifying assumptions (and sometimes lack thereof) associated with commonly used MSE models, and propose a new family of identifying assumptions. While these identifying assumptions may be useful in some applications, \textit{there is no one-size-fits-all solution}. In practice, the use of identifying assumptions should be accompanied by appropriate justification based on the context of the data. However, in some applications none of the identifying assumptions discussed in this article will be appropriate for the data at hand. There is no default identifying assumption that practitioners can fall back on, and so in these scenarios \textit{no estimate of the population size should be produced based on the data at hand}. Such a scenario is clearly unsatisfactory, and thus it is an important task for researchers in the field of MSE to develop new explicit identifying assumptions, so that practitioners are able to select identifying assumptions appropriate for their applications.

\section{Log-Linear and Latent Class Models} \label{sec:common_models}
In this section we describe two commonly used models, which we use to demonstrate the drawbacks of using models that either place unnecessary restrictions on the observed-data distribution or that are not based on explicit identifying assumptions.

\subsection{Log-Linear Models}
\label{sec:loglinear}
For $\bh\in H^*$, let $h_k$ denote the $k$th element of $\bh$. Any set of cell probabilities, $\bpi\in\simplex^{2^K-1}$, can be represented as  $\pi_{\bh}=\mu_{\bh}/\sum_{\bh'\in H}\mu_{\bh'}$, where $\log(\mu_{\bh})=\sum_{\bh'\in H^*}\lambda_{\bh'}\prod_{k=1}^K h_k^{h'_k},$
for some set of log-linear parameters $\blambda=\{\lambda_{\bh}\}_{\bh\in H^*}\in \mathbb{R}^{2^K-1}$.
This leads to the \textit{log-linear parameterization} $\Theta_{LL}=\{N, \blambda \mid N\in\N, \blambda\in \mathbb{R}^{2^K-1}\}$.
Note that under this parameterization, there is no $\lambda_{(0,\cdots,0)}$, so that $\mu_{(0,\cdots,0)}=1$.

For cells in the incomplete table $\bh\in H^*$ such that $\sum_{k=1}^Kh_{k}=1$ we refer to $\lambda_{\bh}$ as a main effect; for $\bh\in H^*$ such that $\sum_{k=1}^Kh_{k}=\ell>1$ we refer to $\lambda_{\bh}$ as an $\ell$-way interaction. The main effects and interactions all have interpretations as log ratios of certain cross-product ratios \citep[see e.g. Chapter 2 of ][]{Bishop_1975}. Of particular interest is the $K$-way, or highest-order, interaction $\lambda_{\bone}$, where $\bone:=(1,\cdots,1)$, for which we have the relationship $\prod_{\bh\in H}\pi_{\bh}^{I_{odd}(\bh)}/\prod_{\bh\in H}\pi_{\bh}^{I_{even}(\bh)}=\exp\{(-1)^{K+1}\lambda_{\bone}\}$, where $I_{odd}(\bh)=I(\sum_{k=1}^K h_k \text{ is odd})$ and $I_{even}(\bh)=I(\sum_{k=1}^K h_k \text{ is even})$, using the convention that $0$ is even. This notation differs from \cite{Bishop_1975} as we index the complete table by $H=\{0,1\}^K$ rather than $\{2,1\}^K$.


The model $\PS_{\Theta_{LL}}$ is equivalent to the unrestricted model $\PS_{\Theta}$, so we need to restrict $\Theta_{LL}$ to identify the unobserved cell probability $\pi_0$. It is standard in this scenario to set $\lambda_{\bone}=0$, so that there is no highest-order interaction in the model. 
Referring to the resulting parameter space as $\Omega_{LL}$, we would like to understand the identifying assumption made by the \textit{saturated model} $\PS_{\Omega_{LL}}$. In Web Appendix C, we show $\PS_{\Omega_{LL}}$ is nonparametric identified and that this no-highest-order interaction (NHOI) assumption corresponds to the explicit identifying assumption $\T(\tilde{\bpi})=(\tilde{\Pi}_{odd}/\tilde{\Pi}_{even})/(1+\tilde{\Pi}_{odd}/\tilde{\Pi}_{even}),$ where $\tilde{\Pi}_{odd}=\prod_{\bh\in H^*}\tilde{\pi}_{\bh}^{I_{odd}(\bh)}$ and $\tilde{\Pi}_{even}=\prod_{\bh\in H^*}\tilde{\pi}_{\bh}^{I_{even}(\bh)}$, which we discuss in more detail in Section \ref{sec:revisit}.


In practice there is an emphasis on achieving low variance estimates of the log-linear parameters and, consequentially, $N$. To this end, rather than just setting the highest-order interaction to zero and using the saturated model, it is common to further restrict the model and set other interactions to zero. This is the case, for example, when restricting to decomposable graphical models \citep{Madigan_1997}, or when only including main effects and 2-way interactions \citep{Silverman_2020}, which can be hard to justify in practice \citep[see e.g. ][]{Dellaportas_1999, Whitehead_2019}. This restricts the observed-data distribution, so that we are making a testable assumption that, in addition to the untestable identifying assumption, must be correct in order for inferences to be valid. The hope is that by specifying a model with fewer parameters, the resulting estimates will have lower variance if the chosen restricted model generated the data. However, if the chosen restricted model did not generate the data, estimates of $N$ can be arbitrarily biased, and more generally can have arbitrarily poor frequentist properties \citep{Regal_1991, Whitehead_2019}.
 
 
This is a classic bias-variance trade off, which has been acknowledged since the seminal work of \cite{Fienberg_1972} (edited to match our notation): ``In analyzing multiple recapture census data our aim is to fit the incomplete $2^K$ table by a log linear model with the fewest possible parameters, since the fewer parameters in an `appropriate' model for estimating $n_0$, the smaller the variance of the estimate. Thus it is not a good practice simply to use the saturated model. On the other hand, if we use a model with too few parameters, we introduce a bias into our estimate of population size that can possibly render the variance formulae of the next section meaningless.'' {Unlike \cite{Fienberg_1972}, we believe there is a clear route to take if one is using the NHOI assumption, in line with our approach described in Section \ref{sec:outlook}: make as few testable assumptions as possible (i.e. use the saturated model $\PS_{\Omega_{LL}}$) in the hopes of not being arbitrarily biased because of incorrect restrictions on the observed data distribution. If one does wish to produce lower variance estimators, we discuss in Web Appendix B how regularization can be used to reduce the variance of estimates, at the cost of increasing the bias of estimates, and some difficulties associated with using regularized estimators.}

\subsection{Latent Class models}
\label{sec:lcm}

Latent class models (LCMs) are typically motivated as models of multivariate categorical data that capture individual heterogeneity when the population can be stratified into $J$ classes, where lists sample individuals independently within each class \citep{Haberman_1979, Manrique-Vallier_2016}. Thus they are  so-called $M_{th}$ models as described in Web Appendix A \citep{Otis_1978}. Corollary 1 of \cite{Dunson_2009} shows that for any set of cell probabilities $\bpi\in\simplex^{2^K-1}$, there exists some $J<\infty$ such that $\bpi$ can be represented as a $J$-class latent class model, i.e. $\pi_{\bh} = \sum_{j=1}^J \nu_j \prod_{k=1}^{K} q_{jk}^{h_k} (1-q_{jk})^{1-h_k},$
where $\bnu=(\nu_1,\cdots,\nu_J)$ are class membership probabilities, and $\bq=\{q_{jk}\}_{j=1,k=1}^{J,K}$ are class specific observation probabilities for each list. This leads to the \textit{latent class model parameterization} $\Theta_{LCM}=\{N, \bnu, \bq, J \mid N\in\N, \bnu\in\simplex^{J-1}, \bq\in(0,1)^{J\times K}, J\in\N \}$. As $\PS_{\Theta_{LCM}}$ is equivalent to the unrestricted model $\PS_{\Theta}$, we need to restrict $\Theta_{LCM}$ to identify the unobserved cell probability $\pi_0$. It is common to fix the number of latent classes, $J$, in advance, to arrive at the the restricted parameterization $\Omega_{LCM, J}=\{N, \bnu, \bq \mid N\in\N, \bnu\in\simplex^{J-1}, \bq\in(0,1)^{J\times K}\}$.

In Web Appendix A we show that $\PS_{\Omega_{LCM,J}}$ is conditionally identified if and only if $2J\leq K$. However, when $\PS_{\Omega_{LCM,J}}$ \textit{is} conditionally identified we do not know what explicit identifying assumption is being made or whether the model is nonparametric identified. 
A recent development in MSE is the use of LCMs with $J$ large enough that $2J>K$ \citep{Manrique-Vallier_2016}. Such LCMs with too many latent classes (i.e. $2J>K$) suffer from the opposite problem of log-linear models: rather than making too many assumptions, and hence restricting the observed-data distribution, so few assumptions are being made that the model is not conditionally identified. In Web Appendix D we show through a variety simulations that this is a practically relevant problem, as we have no guarantees for when estimates based on non-identified models are going to be accurate.

\section{Revisiting Log-Linear Models and Their Identifying Assumptions}
\label{sec:revisit}
In this section we revisit the NHOI  identifying assumption associated with log-linear models and discuss its role in our framing of MSE. We then describe how this assumption can be used as a building block for alternative identifying assumptions.

\subsection{The No-Highest-Order Interaction Assumption}
\label{sec:nhoiinterp}
The NHOI assumption introduced in Section \ref{sec:loglinear} can be interpreted as follows: for any given subset of $K-1$ lists, appearing in all $K-1$ lists is not associated with appearing or not appearing in the $K$th list. Here the meaning of ``associated with" changes as the number of lists $K$ changes. When $K=2$ we are assuming that the odds of appearing in list $1$ conditional on appearing in list $2$ is equal to the odds of appearing in list $1$ conditional on not appearing in list $2$, and thus the lists are independent: $\pi_{(1,0)}/\pi_{(0,0)}=\pi_{(1,1)}/\pi_{(0,1)}$.  When $K=3$ we are assuming that the odds ratio for lists $1$ and $2$ conditional on appearing in list $3$ is equal to the odds ratio for lists $1$ and $2$ conditional on not appearing in list $3$: $\pi_{(1,1,1)}\pi_{(0,0, 1)}/(\pi_{(1,0,1)}\pi_{(0,1,1)})=\pi_{(1,1,0)}\pi_{(0,0, 0)}/(\pi_{(1,0,0)}\pi_{(0,1,0)})$. When $K=4$ we assume that certain ratios of odds ratios are equal, 
and so on for larger $K$.

As discussed in Section \ref{sec:outlook}, in order to use the NHOI assumption in a given application, we need to be able to determine whether or not it is plausible. Odds and odds ratios are commonly used in statistics \citep{Bishop_1975}, and thus the NHOI assumption may be of use when there are $K=2$ or $K=3$ lists. However, higher order measures of association like ratios of odds ratio are more obscure and hard to interpret, which makes the NHOI assumption difficult to use when there are more than $K=3$ lists. This difficulty compounds when considering sensitivity analyses as we explain in the next section. 

\subsection{Sensitivity Analyses for the No-Highest-Order Interaction Assumption}
\label{sec:sensitivity}
Sensitivity analyses aim to gauge how sensitive inferences are to untestable assumptions, and are an important part of missing data workflows \citep[see Chapter 9 of][]{Hogan_2008}. 
The NHOI assumption facilitates sensitivity analyses based on varying the highest-order interaction across a range of non-zero values. In particular, when fixing $\xi=\exp\{(-1)^{K+1}\lambda_{\bone}\}\in\R^+$, 
we show in Web Appendix C that we arrive at the explicit identifying assumption $\T(\tilde{\bpi})=(\tilde{\Pi}_{odd}/\tilde{\Pi}_{even})/(\xi+\tilde{\Pi}_{odd}/\tilde{\Pi}_{even}).$
This generalizes the two list sensitivity analyses of \cite{Lum_2015} and \cite{Gerritse_2015}. 
Under this identifying assumption, rather than assuming certain measures of association are equal, we are assuming one measure is $\xi$ times another. For example, when $K=2$ we are assuming that the odds of appearing in list $1$ conditional on not appearing in list $2$ is $\xi$ times the odds of appearing in list $1$ conditional on appearing in list $2$: $\pi_{(1,0)}/\pi_{(0,0)}=\xi \pi_{(1,1)}/\pi_{(0,1)}$.

In order to perform a meaningful sensitivity analysis, one needs to be able to specify a range of values for the highest-order interaction that are plausible for a given application. Due to our understanding of odds and odds ratios, performing this sort of sensitivity analysis may be possible when there are $K=2$ or $K=3$ lists. When considering more than $K=3$ lists, it can become difficult to even start thinking about whether it is plausible that $\xi$ is less than or greater than $1$, let alone determine specific values of $\xi$ that are plausible. 

\subsection{\texorpdfstring{$K'$}{K'}-List Marginal No-Highest-Order Interaction Assumptions}
\label{sec:marg_ident}
The NHOI assumption can be used as a building block to generate other identifying assumptions. Suppose we can assume that, without loss of generality, the NHOI assumption holds for the first $1<K'<K$ lists, marginal of the remaining $K-K'$ lists. This leads to a new identifying assumption which in general does not imply that there is no highest-order interaction for all $K$ lists. To introduce this assumption formally we need to introduce some notation. Let $G=\{0,1\}^{K'}$ index the marginal $2^{K'}$ contingency table for the first $K'$ lists and $G^*=G\setminus\{0\}^{K'}$. For a set of cell probabilities, $\bpi\in\simplex^{2^K-1}$, and a given cell in the marginal table, $\bg\in G$, let $\pi_{\bg+}=\sum_{\bh\in H}\pi_{\bh}I\{(h_1,\cdots,h_{K'})=\bg\}$ denote the probability of being observed in cell $\bg$ of the marginal table implied by $\bpi$. Similarly let $\tilde{\pi}_{\bg+}=\sum_{\bh \in H^*} \tilde{\pi}_{\bh} I\{(h_1,\cdots,h_{K'})=\bg\}$ and $\tilde{\pi}_{0+}=\sum_{\bh\in H^*}\tilde{\pi}_{\bh}I\{(h_1,\cdots,h_{K'})=(0, \cdots, 0)\}$. 

Assuming that the NHOI assumption holds for the first $1<K'<K$ lists, marginal of the remaining $K-K'$ lists, is equivalent to assuming $\prod_{\bg\in G}\pi_{\bg+}^{I_{odd}(\bg)}/\prod_{\bg\in G}\pi_{\bg+}^{I_{even}(\bg)}=1$. In Web Appendix C we show that this $K'$-list marginal no-highest-order interaction assumption corresponds to the explicit identifying assumption $\T(\tilde{\bpi})=(\tilde{\Pi}_{odd, +}/\tilde{\Pi}_{even, +}- \tilde{\pi}_{0+})/(1+ \tilde{\Pi}_{odd, +}/\tilde{\Pi}_{even, +}-\tilde{\pi}_{0+}),$ where $\tilde{\Pi}_{odd, +}=\prod_{\bg\in G^*}\tilde{\pi}_{\bg+}^{I_{odd}(\bg)}$ and $\tilde{\Pi}_{even, +}=\prod_{\bg\in G^*}\tilde{\pi}_{\bg+}^{I_{even}(\bg)}$. Further, we can perform sensitivity analyses for this assumption by fixing $\prod_{\bg\in G}\pi_{\bg+}^{I_{odd}(\bg)}/\prod_{\bg\in G}\pi_{\bg+}^{I_{even}(\bg)}=\xi\in\R^+$. As we show in Web Appendix C, this leads to the explicit identifying assumption
\begin{equation}
\label{eq:margsens}
    \T(\tilde{\bpi})=\frac{\tilde{\Pi}_{odd, +}/\tilde{\Pi}_{even, +}- \xi\tilde{\pi}_{0+}}{\xi+ (\tilde{\Pi}_{odd, +}/\tilde{\Pi}_{even, +}-\xi\tilde{\pi}_{0+})}.
\end{equation} 
Models that use the assumption that $\prod_{\bg\in G}\pi_{\bg+}^{I_{odd}(\bg)}/\prod_{\bg\in G}\pi_{\bg+}^{I_{even}(\bg)}=\xi\in\R^+$ are not nonparametric identified, as the domain of the identifying assumption is $\tilde{T}=\{\tilde{\bpi}\in\simplex^{2^K-2}\mid \tilde{\Pi}_{odd, +}/(\tilde{\Pi}_{even, +}\tilde{\pi}_{0+})>\xi\}$. 

A special case of this identifying assumption was originally suggested in \cite{Regal_1998} as an alternative to the NHOI assumption. They considered a data set consisting of $K=3$ lists recording cases of spina bifida in upstate New York, 
where they believed that the assumption that two of the lists were marginally independent (i.e., using the $2$-list marginal NHOI assumption) was more plausible than the NHOI assumption. This illustrates that there may be applications where one may be more willing to make marginal assumptions about a subset of $K'$ lists, rather than an assumption involving all $K$ lists. Additionally when there are $K>3$ lists and $K'=2$ or $K'=3$, the $K'$-list marginal NHOI assumption and its sensitivity analyses are much more straightforward to interpret than the highest-order interaction and its sensitivity analyses, as discussed in Sections \ref{sec:nhoiinterp} and \ref{sec:sensitivity}. 

For these reasons, we believe that the $K'$-list marginal NHOI assumption can be useful as an explicit identifying assumption in the toolbox of the MSE practitioner.  However, we emphasize here our message from Section \ref{sec:outlook}: there are no one-size-fits-all identifying assumptions. Specification of identifying assumptions in practice should be accompanied with appropriate justification based on the context of the data. In Section \ref{sec:choice_ident} we attempt to provide such a justification for our use of the $2$-list marginal NHOI assumption in an application estimating the number of civilian casualties in the Kosovo war.

\section{Civilian Casualties in the Kosovo War}
\label{sec:kosovo}
In this section we estimate the number of civilian casualties in the Kosovo war between March 20 and June 22, 1999, using data originally analyzed in \cite{Ball_2002}. The data consist of $K=4$ lists with $n=4400$ observed casualties, and are presented in Table \ref{tab:kosovodat}, reproduced from Section 6 of \protect\cite{Ball_2002}. Three of the lists were constructed from refugee interviews conducted separately by the American Bar Association Central and East European Law Initiative (ABA), Human Rights Watch (HRW), and the Organization for Security and Cooperation in Europe (OSCE). The fourth list was constructed from exhumation reports conducted on behalf of the International Criminal Tribunal for the Former Yugoslavia (EXH). We refer the reader to Appendix 1 of \cite{Ball_2002} for a detailed description of each list. 

\begin{table}[ht]
\centering
\caption{Kosovo dataset, reproduced from Section 6 of \protect\cite{Ball_2002}. 
}
\label{tab:kosovodat}
\begin{tabular}{|l|l|l|l|l|l|}
    \hline
    & ABA & yes & yes & no & no \\ 
    \hline
    & EXH & yes & no & yes & no \\ 
    \hline
    HRW & OSCE &  &  &  &  \\ 
    \hline  
    yes & yes & 27 & 32 & 42 & 123 \\ 
    \hline
    yes & no & 18 & 31 & 106 & 306 \\ 
    \hline
    no & yes & 181 & 217 & 228 & 936 \\ 
    \hline
    no & no & 177 & 845 & 1131 & $n_0$ \\ 
    \hline
\end{tabular}
\end{table}

The Kosovo data was originally analyzed in \cite{Ball_2002} under the NHOI assumption, but as we discuss in the next section, we believe the $K'$-list marginal NHOI assumption is more appropriate. We will analyze the Kosovo data under both assumptions, highlighting the importance of careful specification of the identifying assumption.

\subsection{Choice of Identifying Assumption}
\label{sec:choice_ident}
For our main analysis we will consider two identifying assumptions. 
The first assumption is the $2$-list marginal NHOI assumption described in Section \ref{sec:marg_ident}, where we will assume that the ABA and HRW lists are marginally independent. We believe this assumption is plausible given  that ``there were no overt efforts by any of the researchers to exclude or include witnesses who had participated in another data collection project" \citep[][p. 40]{AAAS_2000} and that the two lists had similarly extensive geographic reach in their interviews. In particular, ABA conducted interviews in Albania, Macedonia, Kosovo, the United States, and Poland, while HRW conducted interviews in Albania, Macedonia, Kosovo, and Montenegro. ABA only conducted around 10\% of its interviews in the United States and Poland, and HRW only conducted 3\% of its interviews in Montenegro. Further, within Kosovo, ABA and HRW conducted interviews in similar geographic regions. For more information on where the lists conducted interviews see Appendix 1 of \cite{Ball_2002}. 

The original analysis of the Kosovo data set in \cite{Ball_2002} used the NHOI assumption described in Section \ref{sec:loglinear}. To justify this assumption for the Kosovo data, as we have $K=4$ lists, we would need to reason about certain ratios of odds ratios being equal, which can be difficult, as discussed in Section \ref{sec:nhoiinterp} and further explained in Web Appendix E.
Nevertheless, we will also analyze the Kosovo data using the NHOI assumption to highlight the importance of careful specification of the identifying assumption.

\subsection{Inference}
\label{sec:ap_inf}
For each identifying assumption, our main analysis will present both a frequentist analysis and a Bayesian analysis, using the methods discussed in Web Appendix B{, to demonstrate how our proposed approach to MSE is agnostic to the inferential framework used}. The Bayesian analysis will use a negative-binomial prior for $N$ and the prior induced for the observed cell probabilities $\tilde{\bpi}$ from using the Dirichlet process prior of \cite{Manrique-Vallier_2016} for the $J$ class LCM $\Omega_{LCM,J}$, with $J=10$ and default hyperparameters, as implemented in the \texttt{R} package \texttt{LCMCR} (see Web Appendix B for further details). In Web Appendix E we perform a prior sensitivity analysis for the Bayesian analyses, exploring the impact of the priors for $N$ and $\tilde{\bpi}$ on our estimates of $N$.

To inform the negative-binomial prior for $N$, we will rely on two studies that attempted to estimate the number of casualties in the Kosovo war using different data sources than \cite{Ball_2002}. \cite{Spiegel_2000} estimated there were $12000$ casualties with a $95\%$ confidence interval of $[5500, 18300]$. \cite{Iacopino_2001} estimated there were $8000$ casualties with a $95\%$ confidence interval of $[5800, 10200]$. 
Using the negative-binomial parameterization given in Table 1 of Web Appendix B, we will use a specification with mean $M=10000$ (the average of the estimates from the two studies) and overdispersion parameter $a=1.6$, which places $95\%$ of the prior mass on $[818, 30371]$. This specification is meant to be weakly informative in the sense that the information it incorporates is intentionally weaker than what is available to us, so as to provide a proper alternative to the ``noninformative" improper scale prior $p(N)\propto 1/N$ discussed in Web Appendix B \citep[see e.g. ][]{Gelman_2017}. 
This prior places mass below the observed sample size of $n=4400$, as we are not attempting to use the observed data to inform our prior. Practically speaking this does not make a difference, as the prior is effectively truncated to $[n, \infty)$ when performing posterior inference.

\subsection{Main Analysis}
In Table \ref{tab:main_analysis_table1} we present the results from our frequentist and Bayesian analyses under the $2$-list marginal NHOI assumption, i.e. assuming marginal independence of the ABA and HRW lists.  {Assuming marginal independence of the ABA and HRW lists, under a frequentist analysis we estimate there were $9691$ civilian casualties, with a $95\%$ confidence interval of $[8074, 11308]$, and under a Bayesian analysis with the chosen priors we estimate there were $9359$ civilian casualties, with a $95\%$ credible interval of $[7967, 11059]$.} These point estimates and uncertainty intervals from these two analyses are in close agreement. Both of the uncertainty intervals include the point estimate from \cite{Iacopino_2001}, but not from \cite{Spiegel_2000}, and fall within the confidence interval of \cite{Spiegel_2000}. Based on the results of the prior sensitivity analysis in Appendix E, the Bayesian analysis is not sensitive to the prior choices for $N$ and $\tilde{\bpi}$.

\begin{table}[ht]
\centering
\caption{Point estimates and  $95\%$ uncertainty intervals for $N$ under the $2$-list marginal NHOI assumption. For the Bayesian analysis the point estimate is the posterior mean.}
\label{tab:main_analysis_table1}
\begin{tabular}{rll}
    \hline
    & Point Estimate & Uncertainty Interval \\ 
    \hline
    Frequentist & 9691& [8074, 11308] \\ 
    Bayesian & 9359 & [7967, 11059] \\
    \hline
\end{tabular}
\end{table}

In Table \ref{tab:main_analysis_table2} we present the results from our frequentist and Bayesian analyses under the NHOI assumption. {Under the NHOI assumption, under a frequentist analysis we estimate there were $16941$ civilian casualties, with a $95\%$ confidence interval of $[5304, 28579]$, and under a Bayesian analysis with the chosen priors we estimate there were $14071$ civilian casualties, with a $95\%$ credible interval of $[9321, 21604]$.} The point estimates and uncertainty intervals from these two analyses are in relative agreement. Both of the uncertainty intervals include the point estimate from \cite{Spiegel_2000}, and the frequentist confidence interval includes the point estimate. Based on the results of the prior sensitivity analysis in Appendix E, the Bayesian analysis is fairly sensitive to the prior choices for $N$ and  $\tilde{\bpi}$.

\begin{table}[ht]
\centering
\caption{Point estimates and  $95\%$ uncertainty intervals for $N$ under the NHOI assumption. For the Bayesian analysis the point estimate is the posterior mean.}
\label{tab:main_analysis_table2}
\begin{tabular}{rll}
  \hline
 & Point Estimate & Uncertainty Interval \\ 
  \hline
    Frequentist & 16941& [5304, 28579] \\ 
    Bayesian &  14071& [9321, 21604]\\
   \hline
\end{tabular}
\end{table}

Focusing on point estimates, we see a large difference between the analyses under the two identifying assumptions (besides the uncertainty interval widths being considerably larger under the NHOI assumption). The point estimates under the NHOI assumption are $75\%$ larger for the frequentist analyses ($50\%$ larger for the Bayesian analyses) than the point estimates under the $2$-list marginal NHOI assumption.  If the $2$-list marginal NHOI assumption truly holds, as we are inclined to believe based on the justification provided in Section \ref{sec:choice_ident}, an analysis based on using the NHOI assumption produces estimates with a large positive bias for the Kosovo data. This should serve as an illustration of the dangers of using the NHOI assumption (or any other identifying assumption) that can not be justified based on the context of the data. If a practitioner can not find an identifying assumption that is appropriate for their data, no estimate of the population size should be produced based on their data, as there is no one-size-fits-all or default identifying assumption to fall back on. There is a need for researchers to develop new explicit identifying assumptions, so that practitioners do not find themselves in such a scenario.

\subsection{A Sensitivity Analysis Probing the \texorpdfstring{$2$}{2}-List Marginal NHOI Assumption}
\label{sec:ap_sens}
While we believe that it is plausible that the ABA and HRW lists are marginally independent, we would also like to understand how sensitive our resulting estimates are to realistic violations of the assumption. If this marginal independence was violated, it would likely be the case that the lists are positively dependent and thus population size estimates under marginal independence are downward biased, as is common in human rights applications \citep[see e.g. the discussion in Section 5 of][]{Lum_2015}. In particular, HRW selected regions in Kosovo to conduct interviews based on reports of human rights violations from refugees and other sources \citep{AAAS_2000}. Thus it seems possible that a casualty appearing in HRW could be more likely to appear in ABA than a casualty that did not appear in HRW. 

We now perform a sensitivity analysis probing the $2$-list marginal NHOI assumption. In Web Appendix E, we provide a similar sensitivity analysis probing the NHOI assumption. 
We will consider models with the identifying assumption \eqref{eq:margsens}, varying $\xi$ over $\{0.7, 0.8, 0.9, 1\}$.
Thus in each case we are assuming that the odds of appearing in ABA conditional on not appearing in HRW is $\xi$ times the odds of appearing in ABA conditional on appearing in HRW, with $\xi=1$ corresponding to the $2$-list marginal NHOI assumption. For each value of $\xi$, we will present both a frequentist analysis and a Bayesian analysis, with the Bayesian analysis using the same priors from the main analysis as presented in Section \ref{sec:ap_inf}.
In Table \ref{tab:sensitivity_analysis_table1} we present the results from our frequentist and Bayesian analyses under each identifying assumption. 

\begin{table}[ht]
\centering
\caption{Point estimates and $95\%$ uncertainty intervals for sensitivity analysis probing the $2$-list marginal NHOI assumption. For the Bayesian analysis the point estimate is the posterior mean. In this table $\xi$ is a marginal odds ratio, as described in Section \ref{sec:marg_ident}.}
\label{tab:sensitivity_analysis_table1}
\begin{tabular}{rlllll}
  \hline
  & $\xi$ = 1 &$\xi$ = 0.9 &$\xi$ = 0.8 &$\xi$ = 0.7 \\ 
  \hline
 Frequentist & 9691 [8074, 11308] & 10534 [8738, 12330] & 11588 [9568, 13607] & 12942 [10636, 15249] \\ 
   Bayesian & 9359 [7967, 11059] & 10155 [8607, 12038] & 11147 [9419, 13258] & 12419 [10451, 14816] \\ 
   \hline
\end{tabular}
\end{table}

The estimates of the number of casualties $N$ increase as the amount of assumed positive dependence  increases, i.e. as $\xi$ decreases, as expected. When $\xi=0.9$, the point estimates and uncertainty intervals are still largely compatible with the point estimates and uncertainty intervals under marginal independence. Thus our estimates under marginal independence are not sensitive to this small amount of positive dependence. However, this is not still the case under stronger positive dependence. When $\xi=0.7$, the uncertainty intervals barely overlap with the uncertainty intervals under marginal independence, and further they do not contain the point estimates under marginal independence. While this may seem like cause for concern, we note that these estimates under stronger positive dependence are still within an order of magnitude of the estimates under independence, and all uncertainty intervals in this sensitivity analysis fall within the confidence interval of \cite{Spiegel_2000}. We note that the frequentist analysis requires a marginal odds ratio of $\xi\approx 0.51$ to produce a point estimate as large as the point estimate under the NHOI assumption. This is a large amount of positive dependence which casts further doubt on the plausibility of the NHOI assumption.

\section{Discussion}
\label{sec:discuss}

{In this article we revisited the framing of MSE as a missing data problem and proposed an approach for MSE that places the identifying assumption front and center in the MSE workflow.} As we have emphasized throughout this article, a natural next step is to develop new explicit identifying assumptions, for situations where the identifying assumptions described in Section \ref{sec:revisit} can not be justified in the context of a given data set.
We believe that this is an extremely under-researched problem that will hopefully gain attention with the re-framing of MSE we present in this article.

The presentation of MSE in this article was focused on estimating the size of a single population. When the population can be stratified based on observed covariates, such as location or time, it may be desirable to estimate the population sizes within each strata. In theory, the methodology developed in this article could be applied independently to each strata. However, stratification can lead to sparse contingency tables, which need significant regularization when estimating $\tilde{\bpi}$. In this case, it would be desirable to develop observed data models that borrow strength across strata.

\newpage
\appendix

\numberwithin{equation}{section}
\numberwithin{lemma}{section}
\numberwithin{theorem}{section}

\section{Web Appendix A: Conditional Identifiability in Models for Heterogeneity}
The purpose of this appendix is to show how common models for heterogeneity fit into the model described in Section 2.2 of the main text, and to provide results regarding conditional identifiability in a particular family of heterogeneous models. The material presented in Appendices \ref{sec:mth_cond_ident}, \ref{sec:lcmthrm}, \ref{sec:lemmaproof}, and \ref{sec:lemmaproof2} previously appeared in the unpublished preprint \cite{Aleshin-Guendel_2020}.

\subsection{Models for Heterogeneity}
Consider the following heterogeneous model
\begin{equation}
	\label{eq:general_model}
	\begin{array}{rl}
		\bpi^i& \stackrel{i.i.d.}{\sim} Q,\\
		\bx_i\mid\bpi^i& \stackrel{ind.}{\sim} \textsc{Categorical}(\bpi^i),
	\end{array}
\end{equation}
where $\bpi^i=\{\pi^i_{\bh}\}_{{\bh}\in H} \in \simplex^{2^K-1}$ for $i=1,\dots,N$. Under this model each individual has its own set of cell probabilities, $\bpi^i$, drawn from some mixing distribution $Q$ on $\mathbb{S}^{2^K-1}$. Working with the heterogeneous model in \eqref{eq:general_model} is equivalent, after marginalizing out $\bpi^i$, to working with the complete-data distribution in Equation (1) of the main text, where $\bpi:=\bpi_Q=E_{Q}(\bpi^i)$ and $E_Q$ denotes the expectation with respect to the mixing distribution $Q$. This is a consequence of the data only providing information about the first moment of the mixing distribution. Suppose $\QQ$ is a family of mixing distributions on $\mathbb{S}^{2^K-1}$. For $Q\in\QQ$, let $\pi_{Q,0}$ denote the induced observed cell probability and $\tilde{\bpi}_Q$ denote the induced observed cell probabilities. The parameter space induced by the family $\QQ$, as a subset of the observed-data parameterization, can then be written as $\Omega_{\QQ}=\{N, \pi_0, \tilde{\bpi}\mid N\in\N, \pi_0=\pi_{Q,0} \text{ and } \tilde{\bpi}=\tilde{\bpi}_Q \text{ for some } Q\in\QQ\}$.

The general heterogeneous model in \eqref{eq:general_model} captures
common models for heterogeneity, including the $M_h$ and $M_{th}$ models \citep{Otis_1978}. The $M_{th}$ model assumes the individual cell probabilities take the form $\pi_{\bh}^i=\prod_{k=1}^K (q_k^i)^{h_k}(1-q_k^i)^{1-h_k}$, where $(q_1^i,\cdots, q_K^i)\stackrel{i.i.d.}{\sim} Q$ and $Q$ is a mixing distribution on $(0,1)^K$. Under this model, conditional on an individual's sampling probabilities, $(q_1^i,\cdots, q_K^i)$, each individual is independently sampled by each list. The $M_h$ model is a submodel of the $M_{th}$ model that assumes that the individual sampling probabilities, $(q_1^i,\cdots, q_K^i)$, are the same for each list, i.e. $q_1^i=\cdots=q_K^i$. Thus the $M_h$ model assumes individuals have the same probability of being sampled by each list. After marginalizing out $\bpi^i$, this enforces a symmetry where the probability of appearing in $k$ lists is the same for each subset of $k$ lists. We do not believe this is plausible in human population settings.

\subsection{Conditional Identifiability in $M_{th}$ Models}
\label{sec:mth_cond_ident}
While there exists a literature characterizing identifiability in $M_{h}$ models \citep{Huggins_2001, Link_2003, Holzmann_2006, Link_2006}, no such results exist for $M_{th}$ models. The purpose of this section is to provide a mechanism for verifying whether the $M_{th}$ model $\PS_{\Omega_{\QQ}}$ is conditionally identifiable based on moments of the mixing distributions $Q\in\QQ$, analogously to the results for $M_h$ models presented in \cite{Holzmann_2006}.

Before proving the main theorem of this section, we have the following lemma, which tells us that for any mixing distribution $Q$ on $(0,1)^K$, the induced cell probabilities, $\bpi_Q$, only depend on $Q$ through its mixed moments. 
\begin{lemma}
	\label{lemma:mixed}
	For any $\bh\in H^*$, $\pi_{Q,\bh}=\sum_{\bh'\in H^*} c_{\bh, \bh'}m_{Q,\bh'}$ where $c_{\bh, \bh'}=(-1)^{\sum_{k=1}^K h'_k - h_k}\prod_{k=1}^K I(h_k \leq h'_k)$ and $m_{Q,\bh'}=E_Q(\prod_{k=1}^K q_k^{h_k'})$.
\end{lemma}
\begin{proof}
	For all $\bh\in H^*$, $\prod_{k=1}^K q_k^{h_k}(1-q_k)^{1-h_k}=
	\sum_{\bh'\in H^*} c_{\bh, \bh'} \prod_{k=1}^K q_k^{h'_k}$ by an application of the multi-binomial theorem (a generalization of the binomial theorem). The result follows from taking the expectation over both sides with respect to $Q$.
\end{proof}

We can restate Lemma \ref{lemma:mixed} in matrix form. Letting  $\bpi_Q^*=(\pi_{Q,\bh})_{\bh\in H^*}$ and $\bmm_Q=(m_{Q,\bh})_{\bh\in H^*}$, we have that $\bpi_Q^*=C\bmm_Q$, where $C=(c_{\bh, \bh'})_{\bh\in H^*, \bh'\in H^*}$. $C$ is invertible as it is upper triangular with non-zero diagonal entries. We are now ready to prove Theorem \ref{theorem:moments}.
\begin{theorem}
	\label{theorem:moments}
	For any two distributions $Q,R$ on $(0,1)^K$, $\tilde{\bpi}_Q=\tilde{\bpi}_R$ is equivalent to $\bmm_Q=A \bmm_R$ for some $A>0$. 
\end{theorem}
\begin{proof}
	$\tilde{\bpi}_Q=\tilde{\bpi}_R$ is equivalent to $\bpi_Q^*/(1-\pi_{Q,0})=\bpi_R^*/(1-\pi_{R,0}).$ Rearranging terms we have that $\bpi_Q^*=\bpi_R^*(1-\pi_{Q,0})/(1-\pi_{R,0}),$ and thus $\bpi_Q^*=A\bpi_R^*$, where $A=(1-\pi_{Q,0})/(1-\pi_{R,0})>0$. Using Lemma \ref{lemma:mixed}, this is equivalent to $C\bmm_Q=AC\bmm_R$, and thus $\bmm_Q=A \bmm_R$ due to the invertibility of $C$. 
\end{proof}

The immediate consequence of Theorem \ref{theorem:moments} is that to verify conditional identifiability of an $M_{th}$ model $\PS_{\Omega_{\QQ}}$, one can demonstrate that if $\bmm_Q=A \bmm_R$ for some $Q,R\in\QQ$, then $\pi_{Q,0}=\pi_{R,0}$. We use this mechanism in the next section to characterize when latent class models (LCMs) are conditionally identifiable.

\subsection{Conditional Identifiability of Latent Class Models}
\label{sec:lcmthrm}
We denote the family of mixing distributions corresponding to LCMs with $J$ classes by $\QQ_J = \{ Q = \sum_{j=1}^J\nu_{Q, j}\prod_{k=1}^K\delta_{q_{Q,jk}} \mid \nu_{Q, j}\geq 0, \sum_{j=1}^J\nu_{Q, j}=1, q_{Q,jk}\in(0,1)^K\}$, so that $\PS_{\Omega_{\QQ_J}}$ is equivalent to $\PS_{\Omega_{LCM,J}}$ from the main text. To provide necessary and sufficient conditions for $\PS_{\Omega_{\QQ_J}}$ to be conditionally identifiable, we restrict the family of mixing distributions to 
$\QQ_J = \{ Q = \sum_{j=1}^J\nu_{Q, j}\prod_{k=1}^K\delta_{q_{Q,jk}} \mid \nu_{Q, j}\geq 0, \sum_{j=1}^J\nu_{Q, j}=1, q_{Q,jk}\in(0,1)^K , q_{Q,jk}\neq q_{Q,j'k} \text{ for } j\neq j'\}$.
This restriction makes the mild assumption that each class' sampling probabilities are distinct, which simplifies the proof of Theorem \ref{thm:lcmthm}. Loosening this restriction could only make the conditions on $J$ for $\QQ_J$ to be identifiable stricter, and thus the conclusions we reach in Section \ref{sec:limits} would still stand for families where this restriction is violated.

There are $J(K+1) - 1$ parameters in $\QQ_J$, thus when $\PS_{\Omega_{\QQ_J}}$ is conditionally identifiable, $J$ satisfies $J(K+1) - 1\leq 2^K -2$, as the observed cell probabilities, $\tilde{\bpi}_Q$, are $2^K -2$ dimensional. However, we now prove that $J$ must satisfy a stricter condition for $\PS_{\Omega_{\QQ_J}}$ to be conditionally identifiable. In Section \ref{sec:limits} we discuss some limitations of this result.

\begin{theorem}
	\label{thm:lcmthm}
	$\PS_{\Omega_{\QQ_J}}$ is conditionally identifiable iff $2J\leq K$.
\end{theorem}
\begin{proof}
	We will first show that if $2J\leq K$, then $\PS_{\Omega_{\QQ_J}}$ is conditionally identifiable. The proof of this direction is similar in spirit to the proofs of Theorem 2 in \cite{Holzmann_2006} and Theorem 1 in \cite{Pezzott_2019}, which were both concerned with characterizing the identifiability of the $M_h$ analogue of $\PS_{\Omega_{\QQ_J}}$. Assume $2J\leq K$, and let $Q,R\in\QQ_J$ such that $\bmm_Q=A \bmm_R$ for some $A>0$, so that we have the following system of equations:
	\begin{equation}
		\label{eq:firstsyst}
		\sum_{j=1}^J\nu_{Q,j}\prod_{k=1}^Kq_{Q,jk}^{h_k} - A
		\sum_{j=1}^J\nu_{R,j}\prod_{k=1}^Kq_{R,jk}^{h_k}=0 \quad (\bh\in H^*).
	\end{equation}
	Let $\I_Q=\{j\mid q_{Q,j}\not\in(q_{R,1},\ldots, q_{R,J})\}$ and $\I_R=\{j\mid q_{R,j}\not\in(q_{Q,1},\ldots, q_{Q,J})\}$, where $q_{Q,j}=(q_{Q,j1},\ldots, q_{Q,jK})$ and $q_{R,j}=(q_{R,j1},\ldots, q_{R,jK})$. We can then rewrite \eqref{eq:firstsyst} as 
	\begin{equation}
		\label{eq:secondsyst}
		\sum_{j=1}^Jy_j\prod_{k=1}^Kq_{Q,jk}^{h_k} - A
		\sum_{i\in \I_R}^J\nu_{R,j}\prod_{k=1}^Kq_{R,jk}^{h_k}=0 \quad (\bh\in H^*),
	\end{equation}
	where $y_j=\nu_{Q,j}$ if $j\in\I_Q$ and $y_j=\nu_{Q,j}-A\nu_{R,j'}$ for some $j'\in\{1,\ldots, J\}\setminus\I_R$ otherwise. Letting $m=|\I_R|=|\I_Q|$ and labelling the elements of $\I_R$ as $i_1,\ldots,i_m$, the system of equations in \eqref{eq:secondsyst} can be written in matrix form as $\Lambda \by=0$, where
	\begin{equation*}
		\Lambda = 
		\begin{pmatrix}
			q_{Q,1K} &  \cdots & q_{Q,JK}&q_{R,i_1K}&\cdots&q_{R,i_mK} \\
			\vdots  & \ddots& \vdots&\vdots  & \ddots& \vdots \\
			\prod_{k=1}^Kq_{Q,1k}^{h_k}&\cdots& \prod_{k=1}^Kq_{Q,Jk}^{h_k}&\prod_{k=1}^Kq_{R,i_1k}^{h_k}&\cdots& \prod_{k=1}^Kq_{R,i_mk}^{h_k} \\
			\vdots& \ddots & \vdots&\vdots  & \ddots& \vdots\\
			\prod_{k=1}^Kq_{Q,1k} &\cdots & \prod_{k=1}^Kq_{Q,Jk}&\prod_{k=1}^Kq_{R,i_1k}&\cdots&\prod_{k=1}^Kq_{R,i_mk}
		\end{pmatrix}, \quad
		\by = 
		\begin{pmatrix}
			y_1  \\
			\vdots  \\
			y_J\\
			-A\nu_{R,i_1}\\
			\vdots\\
			-A\nu_{R,i_m}
		\end{pmatrix},
	\end{equation*}
	and the rows of $\Lambda$ are indexed by $\bh\in H^*$. In Section \ref{sec:lemmaproof}, we prove that $\Lambda$ is full rank, and thus $\by=0$, for any $m \in\{0,\ldots, J\}$. The proof of this direction concludes by examining three possible cases.
	
	\begin{case}
		Suppose $m=0$, i.e. for each $j\in\{1,\ldots, J\}$, there exists some $j'\in\{1,\ldots, J\}$ such that $q_{Q,j}=q_{R,j'}$ and $\nu_{Q,j}=A\nu_{R,j'}$. As $\sum_{j=1}^J \nu_{Q,j}=\sum_{j=1}^J \nu_{R,j}=1$, this implies that $A=1$ and thus $\pi_{Q,0}=\pi_{R,0}$.
	\end{case}
	\begin{case}
		Suppose $m\in\{1,\ldots, J-1\}$, i.e. for each $j\in\{1,\ldots, J\}\setminus\I_Q$, there exists some $j'\in \{1,\ldots, J\}\setminus\I_R$ such that $q_{Q,j}=q_{R,j'}$ and $\nu_{Q,j}=A\nu_{R,j'}$. Further, for each $j\in\I_Q$ and $j'\in \I_R$ $\nu_{Q,j}=\nu_{R,j'}=0$. We can thus ignore the classes $j\in\I_Q$ and $j'\in \I_R$. As $\sum_{j=1}^J \nu_{Q,j}=\sum_{j=1}^J \nu_{R,j}=1$, this implies that $A=1$ and thus $\pi_{Q,0}=\pi_{R,0}$.
	\end{case}
	\begin{case}
		Suppose $m=J$, i.e. for each $j\in\{1,\ldots, J\}$, there exists no $j'\in \{1,\ldots, J\}$ such that $q_{Q,j}=q_{R,j'}$. Then $\nu_{Q,j}=\nu_{R,j}=0$ for $j\in\{1,\ldots, J\}$, which is a contradiction.
	\end{case}
	
	We will now show that if $2J> K$, then $\PS_{\Omega_{\QQ_J}}$ is not conditionally identifiable. To do so we will provide explicit $Q,R\in\QQ_J$ such that $\pi_{Q,0}\neq\pi_{R,0}$, but $\bmm_Q=A \bmm_R$ for $A>0$. This counterexample is modified from \cite{Tahmasebi_2018}, who studied identifiability of families of LCMs outside of the multiple-systems estimation context where $n_{0}$ is observed. Choose $J$ such that $2J> K$. For $j\in\{1,\ldots, J\}$, let $\nu_{Q,j}=\binom{2J}{2j}/(2^{2J-1}-1)$ and $\nu_{R,j}=\binom{2J}{ 2j-1}/(2^{2J-1})$. For $j\in\{1,\ldots, J\}$ and $k\in\{1,\ldots, K\}$, let $q_{Q,jk}=\alpha(2j)$ and $q_{R,jk}=\alpha(2j - 1)$ where $0<\alpha <1 / (2J)$. We thus have that $Q,R\in\QQ_J$, where clearly $Q\neq R$. In Section \ref{sec:lemmaproof2} we prove that for these choices of $Q,R$, $\bmm_Q=A \bmm_R$ for $A>0$ such that $A\neq 1$, and thus $\pi_{Q,0}\neq\pi_{R,0}$.
\end{proof}

\subsection{Proof that $\Lambda$ is Full Rank}
\label{sec:lemmaproof}
We will prove that $\Lambda$ is full rank for any $m\in\{0,\ldots, J\}$ by proving a stronger result. Recall that $K\geq 2$ and let $x_{\ell k}\in(0,1)$ for $\ell\in\{1,\ldots, K\}$ and $k\in\{1,\ldots, K\}$, such that $x_{\ell k}\neq x_{\ell k'}$ for $k\neq k'$. Let 
\begin{equation*}
	X^K = 
	\begin{pmatrix}
		x_{1K} &  \cdots & x_{KK} \\
		\vdots  & \ddots& \vdots \\
		\prod_{k=1}^Kx_{1k}^{h_k}&\cdots& \prod_{k=1}^Kx_{Kk}^{h_k}\\
		\vdots& \ddots & \vdots\\
		\prod_{k=1}^Kx_{1k} &\cdots & \prod_{k=1}^Kx_{Kk}
	\end{pmatrix},
\end{equation*}
where the rows of $X^K$ are indexed by $\bh\in H^*$. We will show that $X^K$ is full rank by induction on $K$. This implies that $\Lambda$ is full rank, as $J+m\leq 2J\leq K$ by assumption for any $m\in\{0,\ldots, J\}$.

For the base case when $K=2$, verifying $X^2$ is full rank is straightforward.
Assume that $X^{K-1}$ is full rank. Let $\bv\in\R^{K\times 1}$ be such that $X^K\bv=0$. For each $\bh\in \{\bh'\in H^*\mid h_K'=0\}$ we have that $v_K\prod_{k=1}^{K-1}x_{Kk}^{h_k}=-\sum_{\ell=1}^{K-1}v_{\ell}\prod_{k=1}^{K-1}x_{\ell k}^{h_k}$, which implies that $\sum_{\ell=1}^{K-1}v_{\ell}(x_{\ell K}-x_{KK})\prod_{k=1}^{K-1}x_{\ell k}^{h_k}=0$. For $\ell\in\{1,\ldots, K-1\}$, let $v_{\ell}'=v_{\ell}(x_{\ell K}-x_{KK})$ and $\bv'=(v_{1}',\ldots, v_{K-1}')$. This leads to the system of equations $X^{K-1}\bv'=0$. By the inductive assumption, $\bv'=0$. Since $x_{\ell K}\neq x_{KK}$ for  $\ell\in\{1,\ldots, K-1\}$, we have that $v_{\ell}=0$ for $\ell\in\{1,\ldots, K-1\}$, and thus $v_K=0$. 

\subsection{Proof of Counterexample}
\label{sec:lemmaproof2}
We will now prove that $m_{Q,\bh}=A m_{R,\bh}$ for all $\bh\in H^*$, where $A=(2^{2J-1})/(2^{2J-1}-1) \neq 1$. Define the function $h(x)=(1-e^{\alpha x})^{2J}=\sum_{i=0}^{2J}\binom{2J}{i}(-1)^ie^{\alpha i x}$. For $t\in\{1,\ldots, K\}$, we can differentiate the series representation of $h$ to find that $h^{(t)}(x)= \sum_{i=0}^{2J}\binom{2J}{i}(-1)^i(\alpha i)^t e^{\alpha i x}$ and thus $h^{(t)}(x)\vert_{x=0}= \sum_{i=0}^{2J}\binom{2J}{ i}(-1)^i(\alpha i)^t=\sum_{i=1}^{2J}\binom{2J}{ i}(-1)^i(\alpha i)^t.$ We can alternatively differentiate the non-series representation of $h$ using the fact that $t\leq K <2J$ and the chain rule for higher order derivatives to find that $h^{(t)}(x)\vert_{x=0}=0$. Let $\bh\in H^*$ and $t=\sum_{k=1}^K h_k\in\{1,\ldots, K\}$. The desired result follows as
\begin{align*}
	m_{Q,\bh}-A m_{R,\bh}&=
	\sum_{j=1}^J\nu_{Q, j} \prod_{k=1}^Kq_{Q,jk}^{h_k}-A\sum_{j=1}^J\nu_{R, j} \prod_{k=1}^Kq_{R,jk}^{h_k} \\
	&=\sum_{j=1}^J \binom{2J}{2j}(2^{2J-1}-1)^{-1}
	\prod_{k=1}^K\{\alpha(2j)\}^{h_k}- A\sum_{j=1}^J
	\binom{2J}{2j-1}(2^{2J-1})^{-1}
	\prod_{k=1}^K \{\alpha(2j - 1)\}^{h_k}\\
	&=(2^{2J-1}-1)^{-1}\sum_{i=1}^{2J} \binom{2J}{ i} (-1)^i(\alpha i)^t=(2^{2J-1}-1)^{-1}\{h^{(t)}(x)\vert_{x=0}\}=0. 
\end{align*} 

\subsection{Limitations of Theorem \ref{thm:lcmthm}}
\label{sec:limits}
Theorem \ref{thm:lcmthm} shows that $\PS_{\Omega_{\QQ_J}}$ is not conditionally identifiable if $2J>K$ by counterexample, by demonstrating two mixing distributions $Q,R\in \QQ_J$ where $\tilde{\bpi}_{Q}=\tilde{\bpi}_{R}$ but $\pi_{Q,0}\neq\pi_{R,0}$. Within each latent class of $Q$ and $R$, the sampling probabilities were the same, meaning $Q$ and $R$ can be seen as mixing distributions of an $M_h$ model. It would be interesting in future work to see whether further restrictions on $\Omega_{\QQ_J}$, for example restrictions not allowing the sampling probabilities within latent classes to be equal, lead to different results concerning conditional identifiability. Another interesting route would be to see whether results concerning \textit{generic identifiability} of latent class models \citep{Allman_2009} could be applied to the multiple-systems estimation setting.

However, this does not mean Theorem \ref{thm:lcmthm} is not a practically useful result. Theorem \ref{thm:lcmthm} provides assumptions under which which we have formal statistical guarantees for when we can estimate the parameters in $\PS_{\Omega_{\QQ_J}}$: the parameters of $\PS_{\Omega_{\QQ_J}}$ can be consistently estimated if $2J\leq K$. When $2J> K$ we currently have no such guarantees. In Web Appendix \ref{sec:lcm_sims} we demonstrate this reality across a variety of simulation studies. 

\section{Web Appendix B: Computation for Conditionally Identified Models}
The purpose of this appendix is to provide details of how computation for conditionally identified models can be carried out in both frequentist and Bayesian frameworks using existing software. Recall from Sections 2.3 and 2.4 of the main text that the complete-data distribution can be written as 
\begin{equation}
	p(\bn, n_0\mid  N, \bpi)= N!\prod_{\bh\in H}\frac{\pi_{\bh}^{n_{\bh}}}{n_{\bh}!}= L_1(N, \pi_0\mid n) L_2(\tilde{\bpi}\mid \bn),
\end{equation}
with $L_1(N, \pi_0\mid n)=\binom{N}{n} \pi_0^{N-n} (1-\pi_0)^{n}$ 
and $L_2(\tilde{\bpi}\mid \bn)= n!\prod_{\bh\in H^*}\tilde{\pi}_{\bh}^{n_{\bh}}/n_{\bh}!$, and that conditionally identified models have parameter spaces of the form $\Omega=\{N, \pi_0, \tilde{\bpi}\mid N\in \N, \pi_0=\T(\tilde{\bpi}), \tilde{\bpi}\in \tilde{S}
\}$. 

\subsection{Computation for Frequentist Multiple-Systems Estimation}
In this section we will first describe an approach for frequentist inference in general conditionally identified models, followed by the specific cases of models using the NHOI and the $K'$-list marginal NHOI identifying assumptions.

\subsubsection{Conditionally Identified Models in General}
\label{sec:freq_general}
Suppose that we are using a conditionally identified model with parameter space $\Omega=\{N, \pi_0, \tilde{\bpi}\mid N\in \N, \pi_0=\T(\tilde{\bpi}), \tilde{\bpi}\in \tilde{S}\}$. Frequentist inference for this general conditionally identified model will follow from the conditional maximum likelihood approach outlined in \cite{Sanathanan_1972} and \cite{Fienberg_1972}. In particular, this approach can be summarized in two steps:
\begin{enumerate}
	\item Estimate the observed cell probabilities $\tilde{\bpi}$ by maximizing the conditional likelihood over the set of possible observed cell probabilities $\tilde{S}$:
	\[
	\hat{\bpi} = \arg\max_{\tilde{\bpi}\in\tilde{S}}L_2(\tilde{\bpi}\mid \bn).
	\]
	\item Estimate the population size $N$ by maximizing the binomial likelihood for $n$ conditional on the estimate of the observed cell probabilities, $\hat{\bpi}$:
	\[
	\hat{N}(\hat{\bpi}) = \arg\max_{N\in\N}L_1(N, \T(\hat{\bpi})\mid n)=\left\lfloor\frac{n}{1-\T(\hat{\bpi})} \right\rfloor,
	\]
	where $\left\lfloor \cdot\right\rfloor$ is the floor function. We will ignore the rounding and write the estimator of $N$ as $\hat{N}(\hat{\bpi}) = n/\{1-\T(\hat{\bpi})\} $. This is well known as the Horvitz-Thompson estimator \citep{Horvitz_1952}.
\end{enumerate}
We note here that $\hat{\bpi}$, and thus $\hat{N}(\hat{\bpi})$, may not exist in general, depending on the set of possible observed cell probabilities $\tilde{S}$. The sample proportions $\{n_{\bh}/n\}_{\bh\in H^*}$ maximize the conditional likelihood over $\simplex^{2^K-2}$, so if the sample proportions lie in $\tilde{S}$, then they maximize the conditional likelihood over $\tilde{S}$. If the sample proportions do not lie in $\tilde{S}$, care must be taken to make sure that $\hat{\bpi}$ exists.

For the rest of this section we will assume that the model is correctly specified, and $\hat{\bpi}$ exists. Let $\tilde{\bpi}^*$ denote the true observed cell probabilities. Suppose it is true, for an estimator $\hat{\bpi}$ of $\tilde{\bpi}^*$, that 
$\sqrt{n}(\hat{\bpi}-\tilde{\bpi}^*)\mid n \to_{d}\textsc{Normal}(0, \Sigma(\tilde{\bpi}^*))$, where $\to_{d}$ denotes convergence in distribution and we are conditioning on $n$ (i.e. ignoring binomial variation in $n$). For example, when the sample proportions $\{n_{\bh}/n\}_{\bh\in H^*}$ lie within $\tilde{S}$, we have that $\hat{\bpi}=\{n_{\bh}/n\}_{\bh\in H^*}$ and $\Sigma(\tilde{\bpi}^*) = \text{diag}(\tilde{\bpi}^*)-\tilde{\bpi}^*(\tilde{\bpi}^*)^{T}$ \citep[see e.g. chapter 14 of][]{Agresti_2003}. For $\tilde{\bpi}\in\tilde{S}$, let $f(\tilde{\bpi})=1/(1-\T(\tilde{\bpi}))$. From the delta method, it follows that $\sqrt{n}(f(\hat{\bpi})-f(\tilde{\bpi}^*))\mid n \to_{d}\textsc{Normal}(0, (\nabla f(\tilde{\bpi}^*))^{T}\Sigma(\tilde{\bpi}^*)\nabla f(\tilde{\bpi}^*))$. Thus for large $n$, $nf(\hat{\bpi})= \hat{N}(\hat{\bpi})\approx\textsc{Normal}(nf(\tilde{\bpi}^*), n(\nabla f(\tilde{\bpi}^*))^{T}\Sigma(\tilde{\bpi}^*)\nabla f(\tilde{\bpi}^*))$. We can then substitute our estimate $\hat{\bpi}$ of the observed cell probabilities for $\tilde{\bpi}^*$, and use this large sample approximation to construct $95\%$ confidence intervals for $N$ of the form $\hat{N}(\hat{\bpi})\pm 1.96 * \sqrt{n(\nabla f(\hat{\bpi}))^{T}\Sigma(\hat{\bpi})\nabla f(\hat{\bpi})}$. The term $(\nabla f(\hat{\bpi}))^{T}\Sigma(\hat{\bpi})\nabla f(\hat{\bpi})$ can be calculated automatically using e.g. the \texttt{delta.method} function in the \texttt{R} package \texttt{msm} \citep{Jackson_2011}.

The confidence interval construction in the last paragraph conditions on $n$, and thus does not incorporate the binomial variation of $n$. Let $N^*$ denote the true population size. For $\tilde{\bpi}\in\tilde{S}$, let $g(\tilde{\bpi})=\T(\tilde{\bpi})/(1-\T(\tilde{\bpi}))$. Following \cite{Fienberg_1972}, unconditional of $n$ we have that $(N^*)^{-1/2}(\hat{N}(\hat{\bpi})-N^*)\to_d\textsc{Normal}(0, g(\tilde{\bpi}^*) + (1-\T(\tilde{\bpi}^*))(\nabla g(\tilde{\bpi}^*))^{T}\Sigma(\tilde{\bpi}^*)\nabla g(\tilde{\bpi}^*))$. Thus for large $N^*$, 
$\hat{N}(\hat{\bpi})\approx\textsc{Normal}(N^*, N^*g(\tilde{\bpi}^*) + N^*(1-\T(\tilde{\bpi}^*))(\nabla g(\tilde{\bpi}^*))^{T}\Sigma(\tilde{\bpi}^*)\nabla g(\tilde{\bpi}^*))$. We can then substitute our estimate $\hat{\bpi}$ of the observed cell probabilities for $\tilde{\bpi}^*$ and our estimate $\hat{N}(\hat{\bpi})$ of the population size for $N^*$, and use this large sample approximation to construct $95\%$  confidence intervals for $N$ of the form $\hat{N}(\hat{\bpi})\pm 1.96 * \sqrt{\hat{N}(\hat{\bpi}) g(\hat{\bpi}) + n(\nabla g(\hat{\bpi}))^{T}\Sigma(\hat{\bpi})\nabla g(\hat{\bpi})}$. Again, the term $(\nabla g(\hat{\bpi}))^{T}\Sigma(\hat{\bpi})\nabla g(\hat{\bpi})$ can be calculated automatically using e.g. the \texttt{delta.method} function in the \texttt{R} package \texttt{msm} \citep{Jackson_2011}.

\subsubsection{Computation for the NHOI and $K'$-List Marginal NHOI Identifying Assumptions}
In this section we will focus on frequentist inference in the specific cases of models using the NHOI and the $K'$-list marginal NHOI identifying assumptions. While one could construct estimators and confidence intervals for $N$, under these assumptions, by hand using the results from the previous section, software is already available which accomplishes these tasks. 

\noindent\textit{NHOI Identifying Assumption}\\
For the NHOI identifying assumption, there are many \texttt{R} packages which produce estimates and confidence intervals for the population size under this assumption. For example, in our Kosovo application we use the \texttt{Rcapture} package \cite{Baillargeon_2007}. The function \texttt{closedpMS.t} produces estimates and standard errors for the population size under all hierarchical log-linear models, including the saturated log-linear model $\PS_{\Omega_{LL}}$. These can then be used to construct confidence intervals for the population size.

\noindent\textit{$K'$-List Marginal NHOI Identifying Assumption}\\
Recall from Section 4.3 of the main text that the $K'$-list marginal NHOI identifying assumption restricts the observed cell probabilities to lie in $\tilde{S}=\{\tilde{\bpi}\in\simplex^{2^K-2}\mid \tilde{\Pi}_{odd, +}/(\tilde{\Pi}_{even, +}\tilde{\pi}_{0+})>1\}$. Thus there are two cases to consider when fitting a model in the frequentist framework using the $K'$-list marginal NHOI identifying assumption:
\begin{enumerate}
	\item The sample proportions $\{n_{\bh}/n\}_{\bh\in H^*}$ lie within $\tilde{S}=\{\tilde{\bpi}\in\simplex^{2^K-2}\mid \tilde{\Pi}_{odd, +}/(\tilde{\Pi}_{even, +}\tilde{\pi}_{0+})>1\}$. 
	\item The sample proportions $\{n_{\bh}/n\}_{\bh\in H^*}$ do not lie within $\tilde{S}$. 
\end{enumerate}
There is a simple way to verify for a given data set, which case one is in. Consider the  restricted data set from just the first $K'$ lists. In particular, using notation from Section 4.3 of the main text, $\{n_{\bg}^{\dagger}\}_{\bg\in G^*}$ is the restricted data, where $n_{\bg}^{\dagger} = \sum_{\bh \in H^*} n_{\bh} I\{(h_1,\cdots,h_{K'})=\bg\}$, and the restricted sample size is $n^{\dagger}=\sum_{\bg\in G^*} n_{\bg}^{\dagger}$. Using this restricted data set of $K'$ lists, one could compute the frequentist population size estimator under the NHOI assumption (for $K'$ lists), using standard software (e.g., the \texttt{Rcapture} package as just described). Call this estimate $\hat{N}^{\dagger}$. Then the sample proportions $\{n_{\bh}/n\}_{\bh\in H^*}$ lie within $\tilde{S}$ as long as $\hat{N}^{\dagger}>n$.

Suppose we are in the second case, i.e.  the sample proportions $\{n_{\bh}/n\}_{\bh\in H^*}$ do not lie in $\tilde{S}$. In this case, $\hat{\bpi}$ may not exist. One needs to verify that $\hat{\bpi}$ exists, and if it does, compute it and derive its asymptotic distribution to compute confidence intervals for $N$ as described  in Appendix \ref{sec:freq_general}. This could potentially be quite difficult technically, so we recommend if one truly believes that the $K'$-list marginal NHOI identifying assumption holds in this case, that they use a Bayesian estimator as described in Appendix \ref{sec:bayes}.

Suppose now we are in the first case, i.e.  the sample proportions $\{n_{\bh}/n\}_{\bh\in H^*}$ lie in $\tilde{S}$. Then $\hat{\bpi}=\{n_{\bh}/n\}_{\bh\in H^*}$, and thus we could then follow the details at the end of Appendix \ref{sec:freq_general} to arrive at a confidence interval for $N$. However, we want to take advantage of existing software in order to compute estimates and confidence intervals for $N$. The following theorem accomplishes this task:
\begin{theorem}
	\label{theorem:equiv}
	Let $\hat{N}$ denote the population size estimator under the $K'$-list marginal NHOI identifying assumption using the full $K$ list data set. Let $\hat{N}^{\dagger}$ denote the population size estimator when restricting to data from just the first $K'$ lists and using the NHOI assumption for $K'$ lists. If the sample proportions $\{n_{\bh}/n\}_{\bh\in H^*}$ lie in $\tilde{S}=\{\tilde{\bpi}\in\simplex^{2^K-2}\mid \tilde{\Pi}_{odd, +}/(\tilde{\Pi}_{even, +}\tilde{\pi}_{0+})>1\}$, then $\hat{N}=\hat{N}^{\dagger}$. 
\end{theorem}
We prove Theorem \ref{theorem:equiv} in Appendix \ref{sec:equivproof}. Theorem \ref{theorem:equiv} tells us that if we want to calculate estimates and confidence intervals for the population size under the $K'$-list marginal NHOI identifying assumption, we can accomplish this by restricting the data set to $K'$ lists, and calculating estimates and confidence intervals for the population size estimate for just these $K'$ lists under the NHOI assumption for $K'$ lists. This can be accomplished using the function \texttt{closedpMS.t} in the \texttt{Rcapture} package.

\noindent\textit{Sensitivity Analyses}\\
\texttt{Rcapture} does not support sensitivity analyses that examine the impact of the NHOI or $K'$-list marginal NHOI identifying assumptions, as described in Section 4.2 and 4.3 of the main text. However, it is straightforward to use the \texttt{glm} function in \texttt{R} to perform these sensitivity analyses, which is what \texttt{Rcapture} uses under the hood. In the code accompanying this manuscript, available at \href{https://github.com/aleshing/central-role-of-identifying-assumptions}{\texttt{github.com/aleshing/central-role-of-identifying-assumptions}}, we provide a function which performs these sensitivity analyses.

\subsubsection{Proof of Theorem \ref{theorem:equiv}}
\label{sec:equivproof}
\begin{proof}
	Suppose we have data from $K$ lists, $\{n_{\bh}\}_{{\bh}\in H^*}$, with observed sample size $n$, and we are using the $K'$-list marginal NHOI assumption, for $1<K'<K$. For this proof, denote the sample proportions by $\tilde{\bpi}=\{n_{\bh}/n\}_{\bh\in H^*}$.
	
	We start by restating some notation from Section 4.3 of the main paper. Let $G=\{0,1\}^{K'}$ index the marginal $2^{K'}$ contingency table for the first $K'$ lists and $G^*=G\setminus\{0\}^{K'}$. Let $\tilde{\pi}_{\bg+}=\sum_{\bh \in H^*} \tilde{\pi}_{\bh} I\{(h_1,\cdots,h_{K'})=\bg\}$ and $\tilde{\pi}_{0+}=\sum_{\bh\in H^*}\tilde{\pi}_{\bh}I\{(h_1,\cdots,h_{K'})=(0, \cdots, 0)\}$.  The $K'$-lists marginal NHOI assumption corresponds to the explicit identifying assumption $\T(\tilde{\bpi})=(\tilde{\Pi}_{odd, +}/\tilde{\Pi}_{even, +}- \tilde{\pi}_{0+})/(1+ \tilde{\Pi}_{odd, +}/\tilde{\Pi}_{even, +}-\tilde{\pi}_{0+}),$ where $\tilde{\Pi}_{odd, +}=\prod_{\bg\in G^*}\tilde{\pi}_{\bg+}^{I_{odd}(\bg)}$ and $\tilde{\Pi}_{even, +}=\prod_{\bg\in G^*}\tilde{\pi}_{\bg+}^{I_{even}(\bg)}$. $\T(\tilde{\bpi})\in(0,1)$ since we assume that $\tilde{\Pi}_{odd, +}/(\tilde{\Pi}_{even, +}\tilde{\pi}_{0+})>1$. 
	
	Now we introduce some new notation. Suppose we are restricted to just the data from the first $K'$ lists. Let $\{n_{\bg}^{\dagger}\}_{\bg\in G^*}$ denote the restricted data, so that $n_{\bg}^{\dagger} = \sum_{\bh \in H^*} n_{\bh} I\{(h_1,\cdots,h_{K'})=\bg\}$, and the restricted sample size is $n^{\dagger}$. Denote the restricted sample proportions by 
	$\tilde{\bpi}^{\dagger}=\{n_{\bg}^{\dagger}/n^{\dagger}\}_{\bg\in G^*}$. Using this restricted $K'$ list data set, the NHOI assumption corresponds to the explicit identifying assumption $\T^{\dagger}(\tilde{\bpi}^{\dagger})=(\tilde{\Pi}_{odd}^{\dagger}/\tilde{\Pi}_{even}^{\dagger})/(1+ \tilde{\Pi}_{odd}^{\dagger}/\tilde{\Pi}_{even}^{\dagger}),$ where $\tilde{\Pi}_{odd }^{\dagger}=\prod_{\bg\in G^*}(\tilde{\pi}_{\bg}^{\dagger})^{I_{odd}(\bg)}$ and $\tilde{\Pi}_{even}^{\dagger}=\prod_{\bg\in G^*}(\tilde{\pi}_{\bg}^{\dagger})^{I_{even}(\bg)}$.
	
	In a frequentist framework, the population size estimate using the $K'$-list marginal NHOI assumption when the estimated observed cell probabilities are $\tilde{\bpi}$ is 
	\[
	\hat{N}=\dfrac{n}{1 - \frac{\tilde{\Pi}_{odd, +}/\tilde{\Pi}_{even, +}- \tilde{\pi}_{0+}}{1+ \tilde{\Pi}_{odd, +}/\tilde{\Pi}_{even, +}-\tilde{\pi}_{0+}}}= 
	n\left[ 1 + \tilde{\Pi}_{odd, +}/\tilde{\Pi}_{even, +}- \tilde{\pi}_{0+}\right].
	\]
	Similarly, in a frequentist framework the population size estimate using the NHOI assumption with the restricted $K'$ list data set is 
	\[
	\hat{N}^{\dagger}=\dfrac{n^{\dagger}}{1 - \frac{\tilde{\Pi}_{odd}^{\dagger}/\tilde{\Pi}_{even}^{\dagger}}{1+ \tilde{\Pi}_{odd}^{\dagger}/\tilde{\Pi}_{even}^{\dagger}}}= 
	n^{\dagger}\left[ 1 + \tilde{\Pi}_{odd}^{\dagger}/\tilde{\Pi}_{even}^{\dagger}\right].
	\]
	Our task is to prove that $\hat{N}=\hat{N}^{\dagger}$.
	
	We list here two useful facts that can be verified through simple algebra:
	\begin{enumerate}
		\item $n^{\dagger}=n-n\tilde{\pi}_{0+}=n[1-\tilde{\pi}_{0+}]$.
		\item $\tilde{\pi}_{\bg}=\tilde{\pi}_{\bg}^{\dagger}[1-\tilde{\pi}_{0+}]$.
	\end{enumerate}
	
	Using the first fact, we can rewrite $\hat{N}^{\dagger}$:
	\begin{align*}
		\hat{N}^{\dagger} &= n^{\dagger}\left[ 1 + \tilde{\Pi}_{odd}^{\dagger}/\tilde{\Pi}_{even}^{\dagger}\right]\\
		&=n[1-\tilde{\pi}_{0+}]\left[ 1 + \tilde{\Pi}_{odd}^{\dagger}/\tilde{\Pi}_{even}^{\dagger}\right]\\
		&=n\left[ 1 + (1-\tilde{\pi}_{0+})(\tilde{\Pi}_{odd}^{\dagger}/\tilde{\Pi}_{even}^{\dagger}) -\tilde{\pi}_{0+}\right].
	\end{align*}
	Thus if we can show that $(1-\tilde{\pi}_{0+})(\tilde{\Pi}_{odd}^{\dagger}/\tilde{\Pi}_{even}^{\dagger})=\tilde{\Pi}_{odd, +}/\tilde{\Pi}_{even, +}$, the proof is complete. Using the second fact, we can rewrite $(1-\tilde{\pi}_{0+})(\tilde{\Pi}_{odd}^{\dagger}/\tilde{\Pi}_{even}^{\dagger})$: 
	\begin{align*}
		(1-\tilde{\pi}_{0+})(\tilde{\Pi}_{odd}^{\dagger}/\tilde{\Pi}_{even}^{\dagger}) &= 
		(1-\tilde{\pi}_{0+})\left[\dfrac{\prod_{\bg\in G^*}(\tilde{\pi}_{\bg}^{\dagger})^{I_{odd}(\bg)}}{\prod_{\bg\in G^*}(\tilde{\pi}_{\bg}^{\dagger})^{I_{even}(\bg)}}\right]\\
		&= 
		(1-\tilde{\pi}_{0+})\left[\dfrac{1-\tilde{\pi}_{0+}}{1-\tilde{\pi}_{0+}}\right]^{2^{K'-1}}\left[\dfrac{\prod_{\bg\in G^*}(\tilde{\pi}_{\bg}^{\dagger})^{I_{odd}(\bg)}}{\prod_{\bg\in G^*}(\tilde{\pi}_{\bg}^{\dagger})^{I_{even}(\bg)}}\right]\\
		&= 
		\left[\dfrac{\prod_{\bg\in G^*}\tilde{\pi}_{\bg+}^{I_{odd}(\bg)}}{\prod_{\bg\in G^*}\tilde{\pi}_{\bg+}^{I_{even}(\bg)}}\right]=\tilde{\Pi}_{odd, +}/\tilde{\Pi}_{even, +}.
	\end{align*}
	
\end{proof}

\subsection{Computation for Bayesian Multiple-Systems Estimation}
\label{sec:bayes}
In this section we will describe a computational approach for Bayesian inference in general conditionally identified models, that allows any prior for the population size, $N$, and any prior for the observed cell probabilities, $\tilde{\bpi}$. Various sensitivity analyses are facilitated from this approach. We further give some guidance to specification of the prior for $N$.

\subsubsection{Bayesian Multiple-Systems Estimation}
\label{sec:bayesmse}
Suppose that we are using a conditionally identified model with parameter space $\Omega=\{N, \pi_0, \tilde{\bpi}\mid N\in \N, \pi_0=\T(\tilde{\bpi}), \tilde{\bpi}\in \tilde{S}\}$, and we have specified independent prior distributions for $N$ and $\tilde{\bpi}$, with densities $p(N)$ and $p(\tilde{\bpi})$. In this section, and the following two sections, we will let $p(\cdot)$ denote a density of a given random variable. The joint posterior of $N$ and $\tilde{\bpi}$ can be written as $p(N, \tilde{\bpi}\mid \bn)\propto L_1(N, \T(\tilde{\bpi})\mid n)L_2(\tilde{\bpi}\mid \bn)p(N) p(\tilde{\bpi})I(\tilde{\bpi}\in \tilde{S}).$
The marginal posteriors of $\tilde{\bpi}$ and $N$ can be written as
\begin{equation}
	\label{eq:marg_dis}
	p(\tilde{\bpi}\mid \bn)\propto p(n\mid \T(\tilde{\bpi}))L_2(\tilde{\bpi}\mid \bn)p(\tilde{\bpi})I(\tilde{\bpi}\in \tilde{S}),
\end{equation}
and $p(N\mid \bn)= \int p(N\mid n, \T(\tilde{\bpi})) p(\tilde{\bpi}\mid \bn) d\tilde{\bpi},$
where $p(n \mid \pi_0) = \sum_{N=n}^{\infty} L_1(N, \pi_0\mid n)p(N)$ and $p(N \mid n, \pi_0) = L_1(N, \pi_0\mid n) p(N) / p(n \mid \pi_0)$, with $\pi_0=\T(\tilde{\bpi})$.  As we discuss in Section \ref{sec:Nprior}, we can compute $p(n\mid \pi_0)$, and thus $p(N \mid n, \pi_0)$, analytically for common priors on $N$. If one has access to Markov chain Monte Carlo (MCMC) samples $\{\tilde{\bpi}^{[t]}\}_{t=1}^T$ from $p(\tilde{\bpi}\mid \bn)$, one can then generate MCMC samples $\{N^{[t]}\}_{t=1}^T$ from $p(N\mid \bn)$ via $N^{[t]}\sim p(N \mid n, \T(\tilde{\bpi}^{[t]}))$. Summaries of the marginal posterior of $N$ can then be calculated based on these samples.

\subsubsection{Mixing and Matching Identifying Assumptions and Priors}
\label{sec:mixmatch}

While computation as described in the previous section may seem straightforward, the marginal posterior for the observed cell probabilities, $p(\tilde{\bpi}\mid \bn)$, depends on the specific combination of priors for $\tilde{\bpi}$ and $N$ and identifying assumption $\T$. Thus we need new MCMC samples from $p(\tilde{\bpi}\mid \bn)$ for each new combination of priors and identifying assumption, which can be difficult both technically and computationally. Rather than develop new MCMC samplers for each combination, we will rely on a combination of existing software and a computationally cheap rejection sampler. 

Let $p_C(\tilde{\bpi}\mid \bn)\propto L_2(\tilde{\bpi}\mid \bn)p(\tilde{\bpi})$ denote the marginal ``posterior" for the observed cell probabilities using just the conditional likelihood $L_2$. We use the subscript $C$ (for ``$C$"onditional) to denote that it is a special density that we are introducing for computational purposes. We can then rewrite the actual marginal posterior for the observed cell probabilities \eqref{eq:marg_dis} as $p(\tilde{\bpi}\mid \bn)\propto p(n\mid \T(\tilde{\bpi}))I(\tilde{\bpi}\in \tilde{S})p_C(\tilde{\bpi}\mid \bn)$. This suggests a computationally cheap rejection sampler to generate samples from $p(\tilde{\bpi}\mid \bn)$, if we have access to MCMC samples from $p_C(\tilde{\bpi}\mid \bn)$ \citep{Smith_1992}:
\begin{enumerate}
	\item Generate $U\sim\textsc{Unif}(0,1)$ and $\tilde{\bpi}\sim p_C(\tilde{\bpi}\mid\bn)$ independently.
	\item If $U<p(n\mid \T(\tilde{\bpi}))I(\tilde{\bpi}\in \tilde{S})/\{\max_{\pi_0}p(n\mid \pi_0)\}$ accept $\tilde{\bpi}$. Else go back to $(1)$.
\end{enumerate}
Thus, for a given prior $p(\tilde{\bpi})$, if we want to perform prior sensitivity analyses for $N$ and/or sensitivity analyses probing the identifying assumption as discussed in Sections 4.2 and 4.3 of the main text, we can take a one time sample from $p_C(\tilde{\bpi}\mid \bn)$, and then reuse this sample to generate samples from $p(\tilde{\bpi} \mid \bn)$ for each combination of prior for $N$ and identifying assumption.
The approach just described is only useful if we have access to MCMC samples from $p_C(\tilde{\bpi}\mid \bn)$. The rest of this section will describe how we can generate samples from the density $p_C(\tilde{\bpi}\mid \bn)$ using existing software.

Previous work in Bayesian MSE specifies priors for $\tilde{\bpi}$ indirectly. In particular, most work specifies priors on reparametrizations of the cell probabilities $\bpi$, such as log-linear models or LCMs, which induce priors for $\bpi$, and thus for $\tilde{\bpi}$. Let $p^w(\bpi)$ denote what we will call the ``working" prior for $\bpi$, which induces the prior $p(\tilde{\bpi})$ we would like to use. We use the superscript $w$ (for ``$w$"orking) to denote that it is a special density that we are introducing for computational purposes. Consider the ``working" posterior for $\bpi$, $p^w(\bpi\mid \bn)\propto \sum_{N=n}^{\infty}p(\bn,n_0\mid N, \bpi)p^w(\bpi)/N$, obtained using the ``working" prior for $N$ of $p^w(N)\propto 1/N$. 
The ``posterior" for $\tilde{\bpi}$ under this working prior combination is equal to $p_C(\tilde{\bpi}\mid \bn)\propto L_2(\tilde{\bpi}\mid \bn)p(\tilde{\bpi})$, as $p(n\mid \pi_0)\propto 1/n$ under the working prior for $N$ (see Table \ref{tab:prior_table}). Thus, given MCMC samples, $\{\bpi^{[t]}\}_{t=1}^T$, drawn from $p^w(\bpi\mid \bn)$, letting $\tilde{\pi}_{\bh}^{[t]}=\pi_{\bh}^{[t]}/(1 - \pi_{0}^{[t]})$, $\{\tilde{\bpi}^{[t]}\}_{t=1}^T$ are MCMC samples drawn from $p_C(\tilde{\bpi}\mid \bn)$.

Thus if we want to use the prior $p(\tilde{\bpi})$ induced by a working prior $p^w(\bpi)$, we can rely on a combination of existing software and a computationally cheap rejection sampler to generate draws from the posterior $p(N, \tilde{\bpi}\mid \bn)$ for any combination of prior for $N$ and identifying assumptions, as long as the software uses the prior $p^w(N)\propto 1/N$. Note that our prior for $N$ \textit{does not} have to be $p^w(N)$. This is the case for most existing software, including the \texttt{R} package \texttt{conting} \citep{Overstall_2014}, which implements a reversible-jump  MCMC sampler to target $p^w(\bpi\mid \bn)$ under a working prior $p^w(\bpi)$ induced by a prior that averages over all hierarchical log-linear models \citep{King_2001}, and the \texttt{R} package \texttt{LCMCR}, which implements a data augementation Gibbs sampler to target $p^w(\bpi\mid \bn)$ under a working prior $p^w(\bpi)$ induced by a Dirichlet process prior for LCMs \citep{Manrique-Vallier_2016}. The steps of the MCMC samplers used in these packages are model specific and we would not be able to use them if we tried to create bespoke MCMC samplers targeting the marginal posterior in \eqref{eq:marg_dis}. We note that this approach is closely related to the working prior approach of \cite{Linero_2017}, with some necessary modifications specific to MSE.

\subsubsection{Recommended Priors for the Population Size, $N$}
\label{sec:Nprior}
In Table \ref{tab:prior_table} we catalog 
$p(n \mid \pi_0) = \sum_{N=n}^{\infty} L_1(N, \pi_0\mid n)p(N)$ and $p(N \mid n, \pi_0) = L_1(N, \pi_0\mid n) p(N) / p(n \mid \pi_0)$ under Poisson, negative-binomial, and binomial priors for $N$, in addition to the class of priors $p(N)\propto (N-\ell)!/N!$,  where $\ell\in\{0, 1,2,\cdots\}$, suggested by \cite{Fienberg_1999}. This class of priors contains both the improper uniform prior, $p(N)\propto 1$, when $\ell=0$, and the improper scale prior, $p(N)\propto 1/N$, when $\ell=1$. If $p(n\mid \pi_0)$ is not available analytically, for example when $p(N)$ is beta-binomial, we recommend truncating the prior for $N$ to the range $\{1,\cdots, N_{max}\}$ where $N_{max}$ is an upper bound on the population size, 
in which case $p(n\mid \pi_0)$ can be computed numerically.

\begin{table}[h]
	\centering
	\caption{Catalog of $p(N\mid n, \pi_0)$ and $p(n\mid \pi_0)$ under common priors for $N$.}
	\label{tab:prior_table}
	\begin{tabular}{l|l|l|l}
		Prior & $p(N)$ & $p(N\mid n, \pi_0)$ & $p(n\mid \pi_0)$ \\ \hline
		$\textsc{Pois}(M)$ & $(M)^Ne^{-M}/N!$ & $n+\textsc{Pois}(\pi_0M)$ &$\textsc{Pois}((1-\pi_0)M)$  \\
		$\textsc{NB}\left(a, \frac{M}{M+a}\right)$ & $\binom{N+a-1}{N}(\frac{M}{M+a})^N(\frac{a}{M+a})^a$  & $n+\textsc{NB}\left(n+a, \frac{M\pi_0}{M+a}\right)$ &$\textsc{NB}\left(a, \frac{(1-\pi_0)M}{(1-\pi_0)M +a}\right)$  \\
		$\textsc{Bin}(M, q)$ & $\binom{M}{N}q^N(1-q)^{M-N}$  & $n+\textsc{Bin}\left(M-n, \frac{\pi_0q}{\pi_0q + 1-q}\right)$ &$\textsc{Bin}(M, (1-\pi_0)q)$  \\
		\small{\cite{Fienberg_1999}} & $\propto (N-\ell)!/N!$  &$n+\textsc{NB}(n-\ell+1, \pi_0)$  & $\propto \frac{(n-\ell)!}{n!}(1-\pi_0)^{\ell-1}$
	\end{tabular}
\end{table}

The improper scale prior, under which $p(\tilde{\bpi} \mid \bn)\propto p_C(\tilde{\bpi} \mid \bn)I(\tilde{\bpi}\in \tilde{S})$, is a common ``noninformative" prior for $N$ and has the nice property that the posterior mean of $N$ conditional on $\tilde{\bpi}$ is the Horvitz-Thompson estimator \citep{Horvitz_1952}, $n/\{1-\T(\tilde{\bpi})\}$, which is well understood in the present context \cite[see e.g.][]{Rukhin_1975}. Recall that the Horvitz-Thompson estimator also arose when considering frequentist inference in Appendix \ref{sec:freq_general}.
Following \cite{Link_2013}, we recommend using this prior in the absence of substantive knowledge about $N$. 

When incorporating substantive knowledge about $N$ into an informative prior for $N$ we recommend using a negative-binomial or beta-binomial prior, as we have found Poisson and binomial priors to be more informative than we would usually like to use. 
For concreteness in the main text we focused on the negative-binomial prior. In Table \ref{tab:prior_table}, we use a common parameterization for the negative-binomial distribution in terms of the mean $M$ and overdispersion parameter $a$. 
This parameterization arises from a Poisson-gamma mixture, where $N\mid \delta\sim\textsc{Poisson}(M\delta)$, $\delta\sim\textsc{Gamma}(a,a)$. As $a\to\infty$ the prior approaches a Poisson prior with mean $M$, and as $a\to0$ the prior approaches the improper scale prior.

\subsection{Regularization and Data Sparsity}
As discussed in Section 3.1 of the main text, when one uses a model that places little to no restrictions on the observed data distribution (as we advocate for in Section 2.6 of the main text), this can lead to population size estimates with large variances associated with them. This typically occurs when the data is sparse, i.e. when some cells of the observed contingency table are small (or even $0$). Data sparsity can be a problem when conducting frequentist analyses, as the standard asymptotic arguments used in Appendix \ref{sec:freq_general} to derive standard errors and confidence intervals are generally not valid. This issue is secondary to our main focus of choosing the identifying assumption, in the sense that the amount of sparsity in the data should not affect the choice of identifying assumption.

We discuss here two possible routes to reduce the variance of population size estimators. The first route is to place restrictions on the observed data distribution, as advocated for by the quote of \cite{Fienberg_1972} in Section 3.1 of the main text. This would require the restricted model to truly hold, otherwise the lower estimated variance would not be valid and the population size estimate could be arbitrarily biased. We would generally prefer not to take this route, as such restrictions are typically hard to justify in practice \citep[see e.g. ][]{Dellaportas_1999, Whitehead_2019}. Further, even if one places correct restrictions on the observed data distribution, in a frequentist analysis the standard errors and confidence intervals derived in Appendix \ref{sec:freq_general} can still be invalid when the data are sparse.

The second route is to use some form of regularization when estimating the observed cell probabilities within a model that places little to no restrictions on the observed data distribution. Regularization reduces the variances of estimates, at the cost of increasing the bias of estimates, by shrinking parameter estimates to a predetermined subset of parameter space. We now briefly discuss how regularization can be incorporated into frequentist or Bayesian analyses:
\begin{itemize}
	\item In a frequentist analysis, regularization can be incorporated through some form of penalized likelihood \citep{Good_1971}, where instead of estimating the observed cell probabilities $\tilde{\bpi}$ by maximizing the conditional likelihood as described in Appendix \ref{sec:freq_general},
	one would maximize the sum of the conditional likelihood and a penalty term
	\begin{equation}
		\label{eq:penalize}
		\hat{\bpi} = \arg\max_{\tilde{\bpi}\in\tilde{S}}L_2(\tilde{\bpi}\mid \bn) - c J(\tilde{\bpi}).
	\end{equation}
	Here $J$ is a penalty function and $c>0$ is a regularization parameter. When $c=0$ the estimate corresponds to the conditional maximum likelihood estimate, and as $c$ increases the estimate gets shrunk to some subset of $\tilde{S}$ defined by the penalty function $J$. It will typically be feasible to obtain estimates by solving \ref{eq:penalize}. However,
	deriving standard errors and confidence intervals for these estimates can be difficult, especially when the data is sparse. Nonstandard asymptotic theory may be required \citep[see e.g.][]{Nardi_2012}. 
	\item In a Bayesian analysis, regularization is inherent due to the prior distribution for $\tilde{\bpi}$. Here the prior serves a similar purpose to the penalty function $J$ in a frequentist analysis, defining the subset of $\tilde{S}$ to which estimates of $\tilde{\bpi}$ are shrunk. Note that the computational techniques in Appendix \ref{sec:bayes} are still valid even when the data are sparse.
\end{itemize}

We note that there is a common difficulty associated with regularizing estimates in a frequentist or Bayesian analysis: choosing where to shrink estimates of the observed cell probabilities, $\tilde{\bpi}$; i.e. choosing the penalty function, $J$, in a frequentist analysis or the prior in a Bayesian analysis. A fruitful direction for future research is to understand what are choices of penalty functions or priors that produce population size estimates with desirable properties when the data are sparse (e.g. good frequentist performance).

\section{Web Appendix C: Identifying Assumption Derivations}
The purpose of this appendix is to derive the identifying assumptions associated with no-highest-order interaction assumption and the $K'$-list marginal no-highest-order interaction assumption.

\subsection{Derivation for No-Highest-Order Interaction Assumption}
Recall from Section 3.1 of the main text that we have the following relationship between the cell probabilities and the highest order interaction, $\lambda_{\bone}$: $\prod_{\bh\in H}\pi_{\bh}^{I_{odd}(\bh)}/\prod_{\bh\in H}\pi_{\bh}^{I_{even}(\bh)}=\exp\{(-1)^{K+1}\lambda_{\bone}\}$, where $I_{odd}(\bh)=I(\sum_{k=1}^K h_k \text{ is odd})$ and $I_{even}(\bh)=I(\sum_{k=1}^K h_k \text{ is even})$. Suppose we fix $\lambda_{\bone}\in\R$, or equivalently $\xi=\exp\{(-1)^{K+1}\lambda_{\bone}\}\in\R^+$.  Under this assumption we have that  $\prod_{\bh\in H}\pi_{\bh}^{I_{odd}(\bh)}/\prod_{\bh\in H}\pi_{\bh}^{I_{even}(\bh)}=\xi$. Multiplying the left-hand side by $1=\left(\frac{1-\pi_0}{1-\pi_0}\right)^{2^{K-1}}$, we find that $\tilde{\Pi}_{odd}/\{[ \pi_0/(1-\pi_0) ]\tilde{\Pi}_{even}\}=\xi$, where $\tilde{\Pi}_{odd}=\prod_{\bh\in H^*}\tilde{\pi}_{\bh}^{I_{odd}(\bh)}$ and $\tilde{\Pi}_{even}=\prod_{\bh\in H^*}\tilde{\pi}_{\bh}^{I_{even}(\bh)}$. Rearranging terms and solving for $\pi_0$, we find that the assumption that $\xi$ is a fixed value corresponds to the explicit functional relationship
\begin{equation}
	\label{eq:sensident}
	\T(\tilde{\bpi})=\frac{\tilde{\Pi}_{odd}/\tilde{\Pi}_{even}}{\xi+\tilde{\Pi}_{odd}/\tilde{\Pi}_{even}}.
\end{equation}
The identifying assumption corresponding to the no-highest-order interaction assumption is recovered by setting $\lambda_{\bone}=0$, or equivalently $\xi=1$: $\T(\tilde{\bpi})=(\tilde{\Pi}_{odd}/\tilde{\Pi}_{even})/(1+\tilde{\Pi}_{odd}/\tilde{\Pi}_{even})$. The observed-data distribution is not restricted by the assumption that the highest-order interaction is fixed, and thus models that use this assumption without any extra assumptions regarding the observed cell probabilities are nonparametric identified.

\subsection{Derivation for $K'$-list Marginal No-Highest-Order Interaction Assumption}
Suppose we assume that $\prod_{\bg\in G}\pi_{\bg+}^{I_{odd}(\bg)}/\prod_{\bg\in G}\pi_{\bg+}^{I_{even}(\bg)}=\xi$, where $\xi\in\R^+$ is fixed. Multiplying the left-hand side by $1=\left(\frac{1-\pi_0}{1-\pi_0}\right)^{2^{K'-1}}$, we find that $\tilde{\Pi}_{odd,+}/\{[ \pi_0/(1-\pi_0) +\tilde{\pi}_{0+}]\tilde{\Pi}_{even,+}\}=\xi$, where $\tilde{\Pi}_{odd, +}=\prod_{\bg\in G^*}\tilde{\pi}_{\bg+}^{I_{odd}(\bg)}$ and $\tilde{\Pi}_{even, +}=\prod_{\bg\in G^*}\tilde{\pi}_{\bg+}^{I_{even}(\bg)}$. Rearranging terms and solving for $\pi_0$, we find that the assumption that $\xi$ is a fixed value corresponds to the explicit functional relationship
\begin{equation}
	\T(\tilde{\bpi})=\frac{\tilde{\Pi}_{odd, +}/\tilde{\Pi}_{even, +}- \xi\tilde{\pi}_{0+}}{\xi+ (\tilde{\Pi}_{odd, +}/\tilde{\Pi}_{even, +}-\xi\tilde{\pi}_{0+})}.
\end{equation} 
The identifying assumption corresponding to the $K'$-list marginal no-highest-order interaction assumption is recovered by setting $\xi=1$: $\T(\tilde{\bpi})=(\tilde{\Pi}_{odd, +}/\tilde{\Pi}_{even, +}- \tilde{\pi}_{0+})/(1+ \tilde{\Pi}_{odd, +}/\tilde{\Pi}_{even, +}-\tilde{\pi}_{0+}).$ 

As noted in Section 4.3 of the main text, the the $K'$-list marginal no-highest-order interaction assumption does not imply that there is no highest-order interaction for all $K$ lists, as $\prod_{\bh\in H}\pi_{\bh}^{I_{odd}(\bh)}/\prod_{\bh\in H}\pi_{\bh}^{I_{even}(\bh)}=(\tilde{\Pi}_{odd}/\tilde{\Pi}_{even})\times(\tilde{\Pi}_{odd, +}/\tilde{\Pi}_{even, +}- \tilde{\pi}_{0+})^{-1}\neq 1$ in general.

\section{Web Appendix D: Latent Class Model Simulations}
\label{sec:lcm_sims}
The purpose of this appendix is conduct simulation studies demonstrating the practical implications of Theorem \ref{thm:lcmthm}. In particular, we present a variety of simulations exploring the frequentist properties of the Bayesian LCM of \cite{Manrique-Vallier_2016}. In each example we generate $200$ data sets from the model in \eqref{eq:general_model} for a given number of lists $K$ and a fixed parameter setting of $\theta\in\Omega_{\QQ_J}$, i.e. a fixed population size $N$ and a $J$-class LCM $Q\in\QQ_J$. For all examples we will use $N\in\{2000, 10000, 100000\}$. For each simulated data set, we fit the Bayesian LCM of \cite{Manrique-Vallier_2016} as implemented in the \texttt{R} package \texttt{LCMCR}, using $J$ latent classes (i.e. the same number that generated the data) and the default prior for $\bnu$, by running the Gibbs sampler implemented in \texttt{LCMCR} for $250,000$ iterations, with the first $50,000$ tossed for burn-in. We note that \texttt{LCMCR} uses the improper scale prior for $N$, i.e. $p(N)\propto1/N$, and a flat prior for $\bq$, i.e. $q_{jk}\stackrel{i.i.d.}{\sim}\textsc{Unif}(0,1)$, which can not be changed. For each parameter setting of $\theta\in\Omega_{\QQ_J}$ we examine the frequentist performance of the posterior median, $95\%$ credible interval, and $50\%$ credible interval for estimating the unobserved cell probability, $\pi_0$, through the sample mean of the posterior medians, the sample coverage of the $95\%$ credible intervals, the sample mean of the $95\%$ credible interval widths over the 200 replications, the sample coverage of the $50\%$ credible intervals, and the sample mean of the $50\%$ credible interval widths over the 200 replications.

\subsection{Example 1}
\label{sec:ex_1}
In this example we consider data from $K=2$ lists generated from the two-class LCM $Q_{1a}$ with parameters given in Table \ref{tab:example_1_table_1}. Under $Q_{1a}$, $\tilde{\pi}_{Q_{1a},(0,1)}=0.276$, $\tilde{\pi}_{Q_{1a},(1,0)}=0.276$, $\tilde{\pi}_{Q_{1a},(1,1)}=0.448$, and $\pi_{Q_{1a},0}=0.316$. There exists another two-class LCM $Q_{1b}$, with parameters given in Table \ref{tab:example_1_table_1}, such that $\tilde{\bpi}_{Q_{1a}}=\tilde{\bpi}_{Q_{1b}}$ but $\pi_{Q_{1b},0}=0.219$. Because  $\PS_{\Omega_{\QQ_2}}$ is not conditionally identified when $K=2$, if we try to perform estimation within $\PS_{\Omega_{\QQ_2}}$, which contains the true data generating model, there is no guarantee that we can estimate well the cell probabilities and population size which generated the data. This example was constructed using the counterexample used to prove Theorem \ref{thm:lcmthm}.

\begin{table}[ht]
	\centering
	\caption{Parameters of two latent class models, $Q_{1a}$ and $Q_{1b}$ (rounded for presentation)} 
	\label{tab:example_1_table_1}
	\begin{tabular}{lllllll}
		& $\nu_{1}$ & $\nu_{2}$ & $q_{11}$ & $q_{12}$ & $q_{21}$ & $q_{22}$ \\
		\hline
		$Q_{1a}$ & 0.500       & 0.500       & 0.248         & 0.248      & 0.743       & 0.743   \\
		$Q_{1b}$ & 0.857     & 0.143     & 0.495          & 0.495          & 0.990           & 0.990           
	\end{tabular}
\end{table}

The results of the simulation using data generated using the LCM $Q_{1a}$ are presented in Table \ref{tab:example_1_table_2}. We see that the posterior median has a negative bias that does not vanish as $N$ increases. One may have thought that the posterior median might possibly be a good estimator for  $\pi_{Q_{1b},0}=0.219$ since $Q_{1a}$ and $Q_{1b}$ induce the same observed-data distribution. However, the posterior median is also negatively biased for estimating $\pi_{Q_{1b},0}$, which suggests there are other LCMs in $\QQ_2$ that induce very similar observed-data distributions to $Q_{1a}$ and $Q_{1b}$ but with different induced unobserved cell probabilities. While the $95\%$ credible interval has nominal coverage when $N=2000$, as $N$ increases, coverage decreases and is no longer nominal. The $50\%$ credible interval have essentially $0$ coverage for settings of $N$, even for $N=2000$ where the $95\%$ credible interval has nominal coverage. This suggests the $95\%$ credible interval only has nominal coverage at $N=2000$ due to wide tails of the posterior for $N$. 

\begin{table}[ht]
	\centering
	\caption{Results of the simulation study where data was generated from the two-class latent class model $Q_{1a}$. Truth is $\pi_{Q_{1a},0}=0.316$.}
	\label{tab:example_1_table_2}
	\begin{tabular}{lllllllll}
		$N$& \shortstack{Mean  \\ Posterior Median} & $95\%$ CI Coverage & \shortstack{Mean  \\ $95\%$ CI Width} & $50\%$ CI Coverage & \shortstack{Mean  \\ $50\%$ CI Width} \\
		\hline
		2000 & 0.148&0.955& 0.332&0.000&0.029 \\
		10000 &  0.146& 0.730& 0.316&0.000&0.023   \\
		100000 &  0.151& 0.265& 0.167&0.055&0.037  
	\end{tabular}
\end{table}

\subsection{Example 2}
One may object to the practicality of Example 1, as it examined a two class LCM constructed using the counterexample from the proof of Theorem \ref{thm:lcmthm}, and is thus an $M_h$ LCM. So we now consider the following example. \cite{Manrique-Vallier_2016} presented a simulation study with $K=5$ lists where data was generated from a LCM with $J=2$ classes, which we reproduce in Table \ref{tab:example_2_table_1}. The parameters of this LCM were based on a hypothetical population where a small proportion of people have a high probability of being observed, and a large proportion of people have a small probability of being observed, which is plausible in some human rights applications. 

\begin{table}[h]
	\caption{Parameters of latent class model which generated data in simulation of \protect\cite{Manrique-Vallier_2016}.}
	\label{tab:example_2_table_1}
	\begin{center}
		\begin{tabular}{ccccccc}
			\hline
			&  & \multicolumn{5}{c}{Sampling probabilities, $\bq$} \\
			\cline{3-7} Class & $\bnu$ & List 1 & List 2 & List 3 & List 4 & List 5 \\ \hline
			1 & 0.900 & 0.033 & 0.033 & 0.099 & 0.132 & 0.033 \\
			2 & 0.100 & 0.660 & 0.825 & 0.759 & 0.990 & 0.693 \\ \hline
		\end{tabular}
	\end{center}
\end{table}

Suppose we only observed lists three and four, so that we have data from $K=2$ lists generated from the two-class LCM $Q_2$ with parameters given in Table \ref{tab:example_2_table_2}. Under $Q_2$, $\pi_{Q_2,0}=0.704$. Just as in the previous example, because $\PS_{\Omega_{\QQ_2}}$ is not conditionally identified when $K=2$, if we try to perform estimation within $\PS_{\Omega_{\QQ_2}}$, which contains the true data generating model, there is no guarantee that we can estimate well the cell probabilities and population size which generated the data. The results of the simulation using data generated using the LCM $Q_{2}$ are presented in Table \ref{tab:example_2_table_3}. We see that the posterior median has a large negative bias that does not vanish as $N$ increases, while the mean $95\%$ and $50\%$ credible interval widths decrease as $N$ increases. Further, the $95\%$ and $50\%$ credible intervals have essentially $0$ coverage across all $N$.

\begin{table}[ht]
	\centering
	\caption{Parameters of latent class model $Q_2$}
	\label{tab:example_2_table_2}
	\begin{tabular}{llllll}
		$\nu_{1}$ & $\nu_{2}$ & $q_{11}$ & $q_{12}$ & $q_{21}$ & $q_{22}$ \\
		\hline
		0.900       & 0.100       & 0.099         & 0.132         & 0.759       & 0.990         
	\end{tabular}
\end{table}

\begin{table}[ht]
	\centering
	\caption{Results of the simulation study where data was generated from the two-class latent class model $Q_{2}$. Truth is $\pi_{Q_2,0}=0.704$.}
	\label{tab:example_2_table_3}
	\begin{tabular}{lllllllll}
		$N$& \shortstack{Mean  \\ Posterior Median} & $95\%$ CI Coverage & \shortstack{Mean  \\ $95\%$ CI Width} & $50\%$ CI Coverage & \shortstack{Mean  \\ $50\%$ CI Width} \\
		\hline
		2000 & 0.285&0.000 & 0.408&0.000&0.055  \\
		10000 &  0.283& 0.010& 0.401&0.000& 0.036 \\
		100000 &  0.285& 0.030& 0.256&0.000&0.035
	\end{tabular}
\end{table}

\subsection{Example 3}
In this example we present two more frequentist simulation studies based on only observing a subset of the five lists from the simulation of \cite{Manrique-Vallier_2016}. 

First suppose that we only observe lists two, three, and four from the simulation of \cite{Manrique-Vallier_2016}, so that we have data from $K=3$ lists generated from the two-class LCM $Q_{3a}$ with parameters given in Table \ref{tab:example_3_table_1}. Under $Q_{3a}$, $\pi_{Q_{3a},0}=0.681$. Because $\PS_{\Omega_{\QQ_2}}$ is not conditionally identified when $K=3$, if we try to perform estimation within $\PS_{\Omega_{\QQ_2}}$, which contains the true data generating model, there is no guarantee that we can estimate well the cell probabilities and population size which generated the data.  The results of the simulation using data generated using the LCM $Q_{3a}$ are presented in Table \ref{tab:example_3_table_2}. We see that the posterior median has a slight negative bias that becomes negligible as $N$ increases. The $95\%$ credible intervals have over-coverage across the different settings of $N$.  The $50\%$ credible intervals have nominal coverage when $N=2000$, but have over-coverage as $N$ increases. 

\begin{table}[h]
	\caption{Parameters of latent class model $Q_{3a}$}
	\label{tab:example_3_table_1}
	\begin{center}
		\begin{tabular}{ccccc}
			\hline
			&  & \multicolumn{3}{c}{Sampling probabilities, $\bq$} \\
			\cline{3-5} Class & $\bnu$ &List 2 & List 3 & List 4 \\ \hline
			1 & 0.900 &  0.033 & 0.099 & 0.132  \\
			2 & 0.100 &  0.825 & 0.759 & 0.990  \\ \hline
		\end{tabular}
	\end{center}
\end{table}

\begin{table}[ht]
	\centering
	\caption{Results of the simulation study where data was generated from the two-class latent class model $Q_{3a}$. Truth is $\pi_{Q_{3a},0}=0.681$.}
	\label{tab:example_3_table_2}
	\begin{tabular}{lllllllll}
		$N$& \shortstack{Mean  \\ Posterior Median} & $95\%$ CI Coverage & \shortstack{Mean  \\ $95\%$ CI Width} & $50\%$ CI Coverage & \shortstack{Mean  \\ $50\%$ CI Width} \\
		\hline
		2000 & 0.622&1.000& 0.339&0.510&0.120\\
		10000 &  0.667& 1.000& 0.274 &0.800&0.091 \\
		100000 &  0.682& 1.000& 0.209&0.965&0.074
	\end{tabular}
\end{table}

Suppose now we only observe lists two, three, four, and five from the simulation of \cite{Manrique-Vallier_2016}, so that we have data from $K=4$ lists generated from the two-class LCM $Q_{3b}$ with parameters given in Table \ref{tab:example_3_table_3}. Under $Q_{3b}$, $\pi_{Q_{3b},0}=0.658$. Because $\PS_{\Omega_{\QQ_2}}$ is conditionally identified when $K=4$, we know that, since $\PS_{\Omega_{\QQ_2}}$contains the true data generating model, we can consistently estimate the cell probabilities and population size which generated the data. The results of the simulation using data generated using the LCM $Q_{3a}$ are presented in Table \ref{tab:example_3_table_4}. We see that the posterior median has a negative bias that becomes negligible as $N$ increases, as expected. The $95\%$ and $50\%$ credible intervals have slight under-coverage when $N=2000$, which becomes nominal as $N$ increases. 

\begin{table}[h]
	\caption{Parameters of latent class model $Q_{3b}$}
	\label{tab:example_3_table_3}
	\begin{center}
		\begin{tabular}{cccccc}
			\hline
			&  & \multicolumn{4}{c}{Sampling probabilities, $\bq$} \\
			\cline{3-6} Class & $\bnu$ & List 2 & List 3 & List 4 & List 5 \\ \hline
			1 & 0.900 &  0.033 & 0.099 & 0.132 & 0.033 \\
			2 & 0.100 & 0.825 & 0.759 & 0.990 & 0.693 \\ \hline
		\end{tabular}
	\end{center}
\end{table}

\begin{table}[ht]
	\centering
	\caption{Results of the simulation study where data was generated from the two-class latent class model $Q_{3b}$. Truth is $\pi_{Q_{3b},0}=0.658$.}
	\label{tab:example_3_table_4}
	\begin{tabular}{lllllllll}
		$N$& \shortstack{Mean  \\ Posterior Median} & $95\%$ CI Coverage & \shortstack{Mean  \\ $95\%$ CI Width} & $50\%$ CI Coverage & \shortstack{Mean  \\ $50\%$ CI Width} \\
		\hline
		2000 & 0.631&0.915&0.190&0.445&0.065\\
		10000 &  0.653& 0.940& 0.089&0.505&0.031  \\
		100000 &  0.658& 0.955& 0.028&0.485&0.010
	\end{tabular}
\end{table}

\subsection{Example 4}
In this example we present three more frequentist simulation studies based on
adding a third class to the LCM from the simulation study of \cite{Manrique-Vallier_2016}, representing a small proportion of the population having a probability of being observed somewhere between the other two classes. The parameters of this new LCM are given in Table \ref{tab:example_4_table_1}.

\begin{table}[h]
	\caption{Parameters of latent class model which generated data in simulation of \protect\cite{Manrique-Vallier_2016}, with a third class added.}
	\label{tab:example_4_table_1}
	\begin{center}
		\begin{tabular}{ccccccc}
			\hline
			&  & \multicolumn{5}{c}{Sampling probabilities, $\bq$} \\
			\cline{3-7} Class & $\bnu$ & List 1 & List 2 & List 3 & List 4 & List 5 \\ \hline
			1 & 0.700 & 0.033 & 0.033 & 0.099 & 0.132 & 0.033 \\
			2 & 0.200 & 0.275 & 0.250 & 0.200 & 0.300 & 0.325 \\
			3 & 0.100 & 0.660 & 0.825 & 0.759 & 0.990 & 0.693 \\ \hline
		\end{tabular}
	\end{center}
\end{table}

First suppose that we only observe lists two, three, and four from the LCM in Table \ref{tab:example_4_table_1}, so that we have data from $K=3$ lists generated from the three-class LCM $Q_{4a}$ with parameters given in Table \ref{tab:example_4_table_2}. Under $Q_{4a}$, $\pi_{Q_{4a},0}=0.613$. Because $\PS_{\Omega_{\QQ_{3}}}$ is not conditionally identified when $K=3$, if we try to perform estimation within $\PS_{\Omega_{\QQ_{3}}}$, which contains the true data generating model, there is no guarantee that we can estimate well the cell probabilities and population size which generated the data.  The results of the simulation using data generated using the LCM $Q_{4a}$ are presented in Table \ref{tab:example_4_table_3}. We see that the posterior median has a negative bias that does not vanish as $N$ increases. The $95\%$ credible intervals have over-coverage across the different settings of $N$, while the $50\%$ credible intervals have under-coverage across the different settings of $N$. Similar to Example 1 in Section \ref{sec:ex_1}, this suggests the $95\%$ credible interval only has over-coverage due to wide tails of the posterior for $N$.

\begin{table}[h]
	\caption{Parameters of latent class model $Q_{4a}$}
	\label{tab:example_4_table_2}
	\begin{center}
		\begin{tabular}{ccccc}
			\hline
			&  & \multicolumn{3}{c}{Sampling probabilities, $\bq$} \\
			\cline{3-5} Class & $\bnu$ & List 1 & List 2 & List 3 \\ \hline
			1 & 0.700  & 0.033 & 0.099 & 0.132  \\
			2 & 0.200 & 0.250 & 0.200 & 0.300 \\
			3 & 0.100  & 0.825 & 0.759 & 0.990  \\ \hline
		\end{tabular}
	\end{center}
\end{table}

\begin{table}[ht]
	\centering
	\caption{Results of the simulation study where data was generated from the two-class latent class model $Q_{4a}$. Truth is $\pi_{Q_{4a},0}=0.613$.}
	\label{tab:example_4_table_3}
	\begin{tabular}{lllllllll}
		$N$& \shortstack{Mean  \\ Posterior Median} & $95\%$ CI Coverage & \shortstack{Mean  \\ $95\%$ CI Width} & $50\%$ CI Coverage & \shortstack{Mean  \\ $50\%$ CI Width} \\
		\hline
		2000 & 0.524&1.000 & 0.387&0.210&0.119 \\
		10000 &  0.537& 1.000& 0.364&0.150&0.102\\
		100000 &  0.538& 1.000& 0.323&0.175&0.096
	\end{tabular}
\end{table}

Next suppose that we only observe lists two, three, four, and five from the LCM in Table \ref{tab:example_4_table_1}, so that we have data from $K=4$ lists generated from the three-class LCM $Q_{4b}$ with parameters given in Table \ref{tab:example_3_table_4}. Under $Q_{4b}$, $\pi_{Q_{4b},0}=0.569$. Because  $\PS_{\Omega_{\QQ_{3}}}$ is not conditionally identified when $K=4$, if we try to perform estimation within $\PS_{\Omega_{\QQ_{3}}}$, which contains the true data generating model, there is no guarantee that we can estimate well the cell probabilities and population size which generated the data. The results of the simulation using data generated using the LCM $Q_{4b}$ are presented in Table \ref{tab:example_4_table_5}. We see that the posterior median has a negative bias that decreases as $N$ increases. While the $95\%$ and $50\%$ credible intervals do not have nominal coverage, coverage improves as $N$ increases (but is still far from nominal even when $N=100000$).

\begin{table}[h]
	\caption{Parameters of latent class model $Q_{4b}$}
	\label{tab:example_4_table_4}
	\begin{center}
		\begin{tabular}{cccccc}
			\hline
			&  & \multicolumn{4}{c}{Sampling probabilities, $\bq$} \\
			\cline{3-6} Class & $\bnu$ & List 1 & List 2 & List 3 & List 4 \\ \hline
			1 & 0.700  & 0.033 & 0.099 & 0.132 & 0.033 \\
			2 & 0.200 & 0.250 & 0.200 & 0.300 & 0.325 \\
			3 & 0.100  & 0.825 & 0.759 & 0.990 & 0.693 \\ \hline
		\end{tabular}
	\end{center}
\end{table}

\begin{table}[ht]
	\centering
	\caption{Results of the simulation study where data was generated from the two-class latent class model $Q_{4b}$. Truth is $\pi_{Q_{4b},0}=0.569$.}
	\label{tab:example_4_table_5}
	\begin{tabular}{lllllllll}
		$N$& \shortstack{Mean  \\ Posterior Median} & $95\%$ CI Coverage & \shortstack{Mean  \\ $95\%$ CI Width} & $50\%$ CI Coverage & \shortstack{Mean  \\ $50\%$ CI Width}  \\
		\hline
		2000 & 0.469&0.525 & 0.199&0.090&0.065 \\
		10000 &  0.509& 0.630& 0.128&0.120& 0.041 \\
		100000 &  0.519& 0.695& 0.066&0.290&0.023
	\end{tabular}
\end{table}

Next suppose that we observe all five lists from the LCM in Table \ref{tab:example_4_table_1}, so that we have data from $K=5$ lists generated from the three-class LCM which we will refer to as $Q_{4c}$. Under $Q_{4c}$, $\pi_{Q_{4c},0}=0.536$. Because $\PS_{\Omega_{\QQ_{3}}}$ is not conditionally identified when $K=5$, if we try to perform estimation within $\PS_{\Omega_{\QQ_{3}}}$, which contains the true data generating model, there is no guarantee that we can estimate well the cell probabilities and population size which generated the data. The results of the simulation using data generated using the LCM $Q_{4c}$ are presented in Table \ref{tab:example_4_table_6}. We see that the posterior median has a negative bias that decreases as $N$ increases. While the $95\%$ and $50\%$ credible intervals do not have nominal coverage, coverage improves as $N$ increases.

\begin{table}[ht]
	\centering
	\caption{Results of the simulation study where data was generated from the two-class latent class model $Q_{4c}$. Truth is $\pi_{Q_{4c},0}=0.536$.}
	\label{tab:example_4_table_6}
	\begin{tabular}{lllllllll}
		$N$& \shortstack{Mean  \\ Posterior Median} & $95\%$ CI Coverage & \shortstack{Mean  \\ $95\%$ CI Width} & $50\%$ CI Coverage & \shortstack{Mean  \\ $50\%$ CI Width} \\
		\hline
		2000 & 0.435&0.415 & 0.169&0.050&0.055 \\
		10000 &  0.500& 0.875& 0.132&0.370&0.044\\
		100000 &  0.523& 0.895& 0.053&0.490&0.018
	\end{tabular}
\end{table}

\subsection{Takeaways}
When using the model $\PS_{\Omega_{\QQ_J}}$ for multiple-systems estimation, one is relying on the assumption that the data was generated from a distribution in $\PS_{\Omega_{\QQ_J}}$. If a practitioner is comfortable with the assumption that $2J\leq K$, then we know the model is conditionally identified, and thus this assumption is a combination of an explicit identifying assumption (which is currently unknown) and possibly some restrictions on the observed-data distribution. Due to conditional identification, the practitioners  have guarantees under this assumption that they can estimate the population size, and other parameters, well if their observed sample size $n$ is large enough. However, if a practitioner is not comfortable with this assumption, and chooses to use $J>K/2$, they have no such guarantees as they are using a model that is not conditionally identified.

Through the four example simulation studies in this appendix we saw examples where models that were not conditionally identified had good frequentist performance ($Q_{3a}$, $Q_{4a}$) and bad frequentist performance ($Q_{1}$, $Q_{2}$, $Q_{4b}$, $Q_{4c}$) according to some of our simulation summary measures. The good and bad frequentist performances could have been due to 
\begin{itemize}
	\item where the prior of \cite{Manrique-Vallier_2016} places mass in the parameter space $\Omega_{\QQ_J}$ (e.g. good frequentist performance if it places enough prior mass around the true data generating parameters),
	\item whether there actually exists other LCMs in $\QQ_J$ that induce similar observed cell probabilities to the true data generating parameters but a different unobserved cell probability (e.g. good frequentist performance if other LCMs do not exist with these properties),
	\item or some combination of the two previous factors.
\end{itemize}
We currently have no way to tease apart these factors and tell when a model that is not conditionally identified will have good or bad performance. This is a problem for using these models in practice, as we have no way to tell practitioners ``under these assumptions the model will perform well". 

We believe there are two routes forward to combat this problem, if one wants to use LCMs for multiple-systems estimation. The first option is to further study technical results for conditional identification in LCMs. For example, as we discussed in Section \ref{sec:limits}, suppose we can prove under further (practically relevant) restrictions on $\QQ_J$ that $\PS_{\Omega_{\QQ_J}}$ is conditionally identified for some $J>K/2$. We would then be able to expand the range of models we could fit under which we had guarantees that we could estimate well the parameters of the model. 

The other option is to study LCMs through the framework of partial identification \citep{Tamer_2010, Gustafson_2010}, which was recently used in multiple-systems estimation by \cite{Sun_2020} for frequentist inference for partially-identified log-linear models. This would require both: 1) a better technical understanding of what parameters, or functions of parameters, of LCMs are not identified, and 2) placing substantively meaningful priors on the non-identified parameters (i.e. priors informed by substantive knowledge concerning the population of interest and how the data was collected) if a Bayesian approach is taken. Without 1), the best we can do in a Bayesian approach is to place substantively meaningful priors on all LCM parameters, i.e. on $\Omega_{\QQ_J}$. The prior for $\Omega_{\QQ_J}$ of \cite{Manrique-Vallier_2016} is based on the Dirichlet Process prior specification of \cite{Dunson_2009}, which is a prior of technical convenience. Specifying a substantively meaningful prior for $\Omega_{\QQ_J}$ would require being able to specify a prior for the class membership probabilities $\bnu$ and for the class specific observation probabilities $\bq$. It is difficult to imagine a scenario in which a practicioner would have knowledge of the population of interest and how the data was collected that could be incorporated into priors for all $J(K+1) - 1$ of the parameters ($\bnu$ and $\bq$).

While we do not believe that latent class models cannot be used for multiple-systems estimation ({see our application in Section 5 of the main text where we use the LCM prior of \cite{Manrique-Vallier_2016} to induce a prior for the observed cell probabilities $\tilde{\bpi}$}), we do believe that there needs to be further research to understand under what assumptions LCMs do and do not perform well in practice. We discuss one further area of research before concluding this section. The start of this section began by assuming that a practitioner assumed their data was generated by a distribution in  $\PS_{\Omega_{\QQ_J}}$. It is not clear to the authors how in practice one would choose a specific value of $J$. In practice, how would a practitioner choose between $\PS_{\Omega_{\QQ_J}}$ and $\PS_{\Omega_{\QQ_{J'}}}$ for $J\neq J'$? What characteristics of the population being studied and the data collection process would allow one to differentiate between these two models? Research into understanding how to elicit plausible values of $J$ would help to justify the use of the model $\PS_{\Omega_{\QQ_J}}$ in practice.

\section{Web Appendix E: Kosovo Analysis Appendix}

This appendix serves three purposes: 1) to describe the difficulty in justifying the NHOI assumption for the Kosovo data, 2) to describe a prior sensitivity analysis for the Bayesian analyses of the Kosovo data, and 3) to describe a sensitivity analysis for the Kosovo data probing the NHOI assumption.

\subsection{The No-Highest-Order Interaction Assumption}
\label{sec:nhoi_issue}
The Kosovo data set has $K=4$ lists, which we will order (without loss of generality) so that the American Bar Association Central and East European Law Initiative (ABA) list is first, the Human Rights Watch (HRW) list is second, the Organization for Security and Cooperation in Europe (OSCE) list is third, and the list constructed from exhumation reports conducted on behalf of the International Criminal Tribunal for the Former Yugoslavia (EXH) is fourth. Let $\text{Odds}(h_1=1\mid h_2=1, h_3,h_4)=\pi_{(1,1,h_3,h_4)}/\pi_{(0,1,h_3,h_4)}$ denote the odds that an individual is observed in list 1, conditional on being observed in list 2 and the inclusion patterns $h_3, h_4$ for lists 3 and 4. For example, if $h_3=0$ and $h_4=1$, $\text{Odds}(h_1=1\mid h_2=1, h_3=0,h_4=1)$ is the odds that an individual is observed in list 1, conditional on being observed in lists 2 and 4 and not being observed in list 3. Similarly let $\text{Odds}(h_1=1\mid h_2=0, h_3,h_4)=\pi_{(1,0,h_3,h_4)}/\pi_{(0,0,h_3,h_4)}$ denote the odds that an individual is observed in list 1, conditional on not being observed in list 2 and the inclusion patterns $h_3, h_4$ for lists 3 and 4. We can then define $\text{OR}(h_3,h_4)=\text{Odds}(h_1=1\mid h_2=1, h_3,h_4)/\text{Odds}(h_1=1\mid h_2=0, h_3,h_4)$ as the odds ratio for lists 1 and 2, conditional on the inclusion patterns $h_3, h_4$ for lists 3 and 4. Following Section 4.1 of the main text, the no-highest-order interaction assumption assumes that $\text{OR}(1,0)/\text{OR}(0,0)=\text{OR}(1,1)/\text{OR}(0,1)$, i.e. the highest-order interaction for the first three lists, conditional on not being observed in list 4, $\text{OR}(1,0)/\text{OR}(0,0)$,  is equal to the highest-order interaction for the first three lists, conditional on being observed in list 4, $\text{OR}(1,1)/\text{OR}(0,1)$.

This assumption is obscure and hard to justify based on our knowledge of how the four lists were generated. As the validity of our analysis rests on this assumption being correct, we stress that we are not confident that this assumption holds, and thus we are not confident in the validity of the analysis of the Kosovo data set using the NHOI assumption.

\subsection{Prior Sensitivity Analyses}
In this section we perform prior sensitivity analyses for the Bayesian analyses of the Kosovo data from the main text. For $N$, we will consider the negative-binomial prior specification described in the main text, in addition to the improper scale prior discussed in Appendix \ref{sec:Nprior}.
For the observed cell probabilities $\tilde{\bpi}$, we will consider four prior specifications : 1) the prior induced from using the Dirichlet process prior of \cite{Manrique-Vallier_2016} for the $J$ class LCM $\Omega_{LCM,J}$, with $J=10$ and default hyperparameters, as implemented in the \texttt{R} package \texttt{LCMCR} (i.e. the prior used in the main analyses), 2) a flat Dirichlet prior, i.e. $\tilde{\bpi}\sim\textsc{Dirichlet}(1,\cdots,1)$, 
3) the prior induced from using $\textsc{Normal}(0, 5^2)$ priors for the log-linear parameters in the saturated log-linear model $\Omega_{LL}$, fit using the \texttt{Stan} probabilistic programming language \citep{Carpenter_2017}, and 4) the prior induced from using the Bayesian model averaging prior of \cite{King_2001} for the log-linear parameters in the saturated log-linear model $\Omega_{LL}$, with the unit information prior on log-linear parameters, as implemented in the \texttt{R} package \texttt{conting} \citep{Overstall_2014}. We note that \texttt{conting} uses an alternative log-linear parameterization based on sum to zero constraints rather than corner point constraints used in Section 3.1 of the main text. For each combination of identifying assumption and priors for $N$ and $\tilde{\bpi}$ we fit the corresponding model using the computational approach described in Appendix \ref{sec:mixmatch}. 

In Table \ref{tab:prior_sensitivity_analysis_table1} we present posterior means and $95\%$ credible intervals for $N$ under each prior combination under the $2$-list marginal NHOI assumption, i.e. assuming marginal independence of the ABA and HRW lists. The posterior density for $N$ under each prior combination under the $2$-list marginal NHOI assumption is displayed in Figure \ref{fig:prior_sensitivity_analysis_plot1}. For each prior for $\tilde{\bpi}$, the posterior for $N$ does not appear to be sensitive to the prior for $N$, as the point estimates and credible intervals are essentially the same between the two priors for $N$. Across the different priors for $\tilde{\bpi}$, the posterior summaries are fairly consistent, with the posterior summaries under the LCM prior for $\tilde{\bpi}$ being slightly lower than under the other priors. We note that all of the credible intervals fall within the confidence interval of \cite{Spiegel_2000}.

\begin{table}[ht]
	\centering
	\caption{Posterior means and $95\%$ credible intervals for $N$ under each combination of prior for $N$ and $\tilde{\bpi}$, under the $2$-list marginal NHOI assumption. }
	\label{tab:prior_sensitivity_analysis_table1}
	\begin{tabular}{rll}
		\hline
		& Improper Scale Prior & Negative-Binomial \\ 
		\hline
		Conting & 9618 [8224, 11195] & 9621 [8232, 11191] \\ 
		Dirichlet & 9536 [8113, 11252] & 9540 [8123, 11247] \\ 
		LCMCR & 9353 [7959, 11063] & 9359 [7967, 11059] \\ 
		Log-Linear & 9764 [8277, 11549] & 9766 [8288, 11550] \\ 
		\hline
	\end{tabular}
\end{table}
\begin{figure}[ht]
	\centering
	\includegraphics[width=1\linewidth]{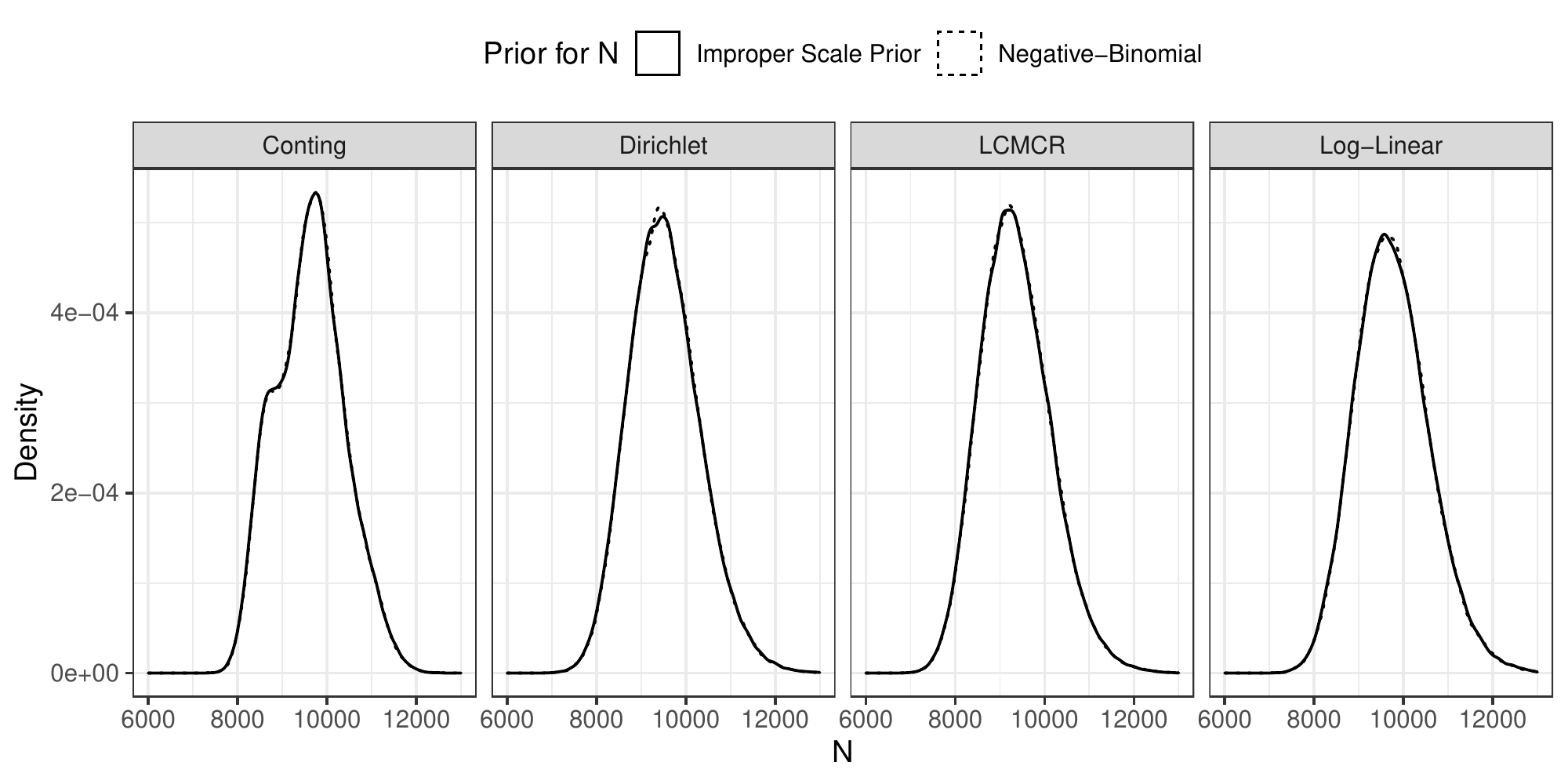}
	\begin{minipage}[b]{0.95\textwidth}
		\caption{Posterior density of $N$ under each combination of prior for $N$ and $\tilde{\bpi}$, under the $2$-list marginal NHOI assumption.}
		\label{fig:prior_sensitivity_analysis_plot1}
	\end{minipage} 
\end{figure}

In Table \ref{tab:prior_sensitivity_analysis_table2} we present posterior means and $95\%$ credible intervals for $N$ under each prior combination under the NHOI assumption. The posterior density for $N$ under each prior combination under the NHOI assumption is displayed in Figure \ref{fig:prior_sensitivity_analysis_plot2}. For each prior for $\tilde{\bpi}$, the posterior for $N$ is somewhat sensitive to the prior for $N$, as the posterior mean and credible interval limits are always larger under the improper scale prior compared to the negative-binomial prior for $N$. The posterior for $N$ appears to be the most sensitive to the prior for $N$ under the Dirichlet and log-linear priors for $\tilde{\bpi}$, where the posterior means and upper credible interval limits increase by several thousand when using the improper scale prior for $N$ instead of the negative-binomial prior. 
Across the different priors for $\tilde{\bpi}$, the posteriors corresponding to the Dirichlet prior, the log-linear model prior, and the LCM prior of \cite{Manrique-Vallier_2016} are in relative agreement. The posterior corresponding to the Dirichlet prior is the most diffuse of the three, and the posterior corresponding to the LCM prior of \cite{Manrique-Vallier_2016} is the most concentrated of the three. The posterior corresponding to the log-linear model prior of \cite{King_2001}, implemented in the \texttt{conting} package, is multimodal, which is not unexpected as it is performing Bayesian model averaging \citep{Hoeting_1999} over all hierarchical log-linear models. Due to this multimodality, point estimates (e.g. the posterior mean) may not be reliable summaries of the posterior distribution. We note that all of the credible intervals contain the point estimate of \cite{Spiegel_2000}.

\begin{table}[ht]
	\centering
	\caption{Posterior means and $95\%$ credible intervals for $N$ under each combination of prior for $N$ and $\tilde{\bpi}$, under the  NHOI assumption. }
	\label{tab:prior_sensitivity_analysis_table2}
	\begin{tabular}{rll}
		\hline
		& Improper Scale Prior & Negative-Binomial \\ 
		\hline
		Conting & 13000 [9202, 19971] & 12694 [9175, 19299] \\ 
		Dirichlet & 18500 [9402, 35908] & 16051 [9098, 27679] \\ 
		LCMCR & 14695 [9423, 23675] & 14071 [9321, 21604] \\ 
		Log-Linear & 16209 [8731, 30025] & 14719 [8579, 24878] \\ 
		\hline
	\end{tabular}
\end{table}
\begin{figure}[ht]
	\centering
	\includegraphics[width=1\linewidth]{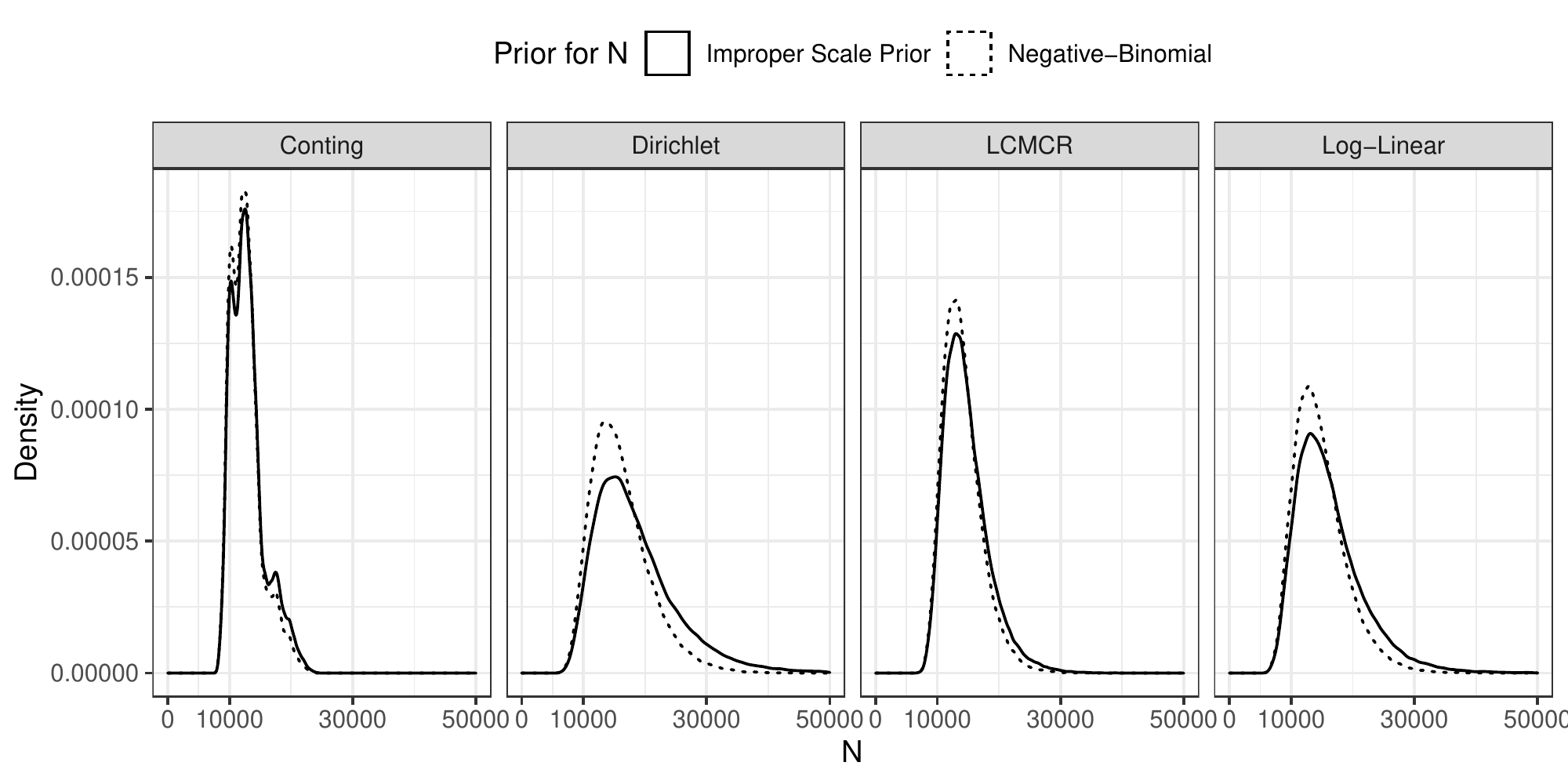}
	\begin{minipage}[b]{0.95\textwidth}
		\caption{Posterior density of $N$ under each combination of prior for $N$ and $\tilde{\bpi}$, under the  NHOI assumption.}
		\label{fig:prior_sensitivity_analysis_plot2}
	\end{minipage} 
\end{figure}

\subsection{A Sensitivity Analysis Probing the NHOI Assumption}
We now perform a sensitivity analysis probing the no-highest-order interaction assumption. We will consider models with the identifying assumption in Section 4.2 of the main text, varying $\xi$ over $\{1/2, 2/3, 1, 3/2, 2\}$ \citep[following][]{Gerritse_2015}. For each value of $\xi$, we will present both a frequentist analysis and a Bayesian analysis, with the Bayesian analysis using the same priors from the main analysis as presented in Section 5.1 of the main text. This sensitivity analysis is limited in that we followed \cite{Gerritse_2015} and chose an arbitrary range of values for $\xi$ around 1. Due to the difficulty in interpreting the highest-order interaction when there are $K=4$ lists, we are not able to say with confidence whether this range of values is meaningful or not. 
In Table \ref{tab:sensitivity_analysis_table2} we present the results from our frequentist and Bayesian analyses under each identifying assumption.

\begin{table}[ht]
	\centering
	\caption{Point estimates and $95\%$ uncertainty intervals for sensitivity analysis probing the NHOI assumption. For the Bayesian analysis the point estimate is the posterior mean. In this table $\xi$ is a ratio of ratios of odds ratios, as described in Section 4.2 of the main text and Appendix \ref{sec:nhoi_issue}.}
	\label{tab:sensitivity_analysis_table2}
	\begin{tabular}{rllllll}
		\hline
		&$\xi$ = 1 / 2 &$\xi$ = 2 / 3 &$\xi$ = 1 &$\xi$ = 3 / 2 &$\xi$ = 2 \\ 
		\hline
		Frequentist & 29483 [6210, 52757] & 23212 [5757, 40668] & 16941 [5304, 28579] & 12761 [5002, 20520] & 10670 [4851, 16490] \\ 
		Bayesian & 21476 [13518, 33507] & 17983 [11492, 27987] & 14071 [9321, 21604] & 11121 [7766, 16564] & 9538 [6943, 13821] \\ 
		\hline
	\end{tabular}
\end{table}

The results are not very robust to misspecification of $\xi$ in the chosen range. The uncertainty intervals when $\xi=1/2$ and $\xi=2$ barely overlap. For the Bayesian analysis, the posterior mean when $\xi=2$ is $32\%$ lower than the posterior mean when $\xi=1$ (i.e. under the no-highest-order interaction assumption), the posterior mean when $\xi=1/2$ is $53\%$ higher than the posterior mean when $\xi=1$, and the posterior mean $\xi=1/2$ is more than twice the posterior mean when $\xi=2$. These differences are even more dramatic for the frequentist analysis. This lack of robustness to misspecification of $\xi$ would be a cause for concern if the no-highest-order interaction assumption was plausible, and the deviations from the assumption in terms of $\xi$ were also plausible, in the context of the Kosovo data set. 

\bibliographystyle{biom} 
\bibliography{crc_refs_supp.bib}

\end{document}